\documentclass[manuscript,screen]{acmart}

\AtBeginDocument{%
  }

\setcopyright{acmlicensed}
\copyrightyear{2026}
\acmYear{2026}
\acmDOI{XXXXXXX.XXXXXXX}

\title{Privacy Against Agnostic Inference Attacks in Vertical Federated Learning}

\author{Morteza Varasteh}
\email{m.varasteh@essex.ac.uk}
\affiliation{
  \institution{School of Computer Science and Electronic Engineering, University of Essex}
  \city{Colchester CO4 3SQ}
  \country{UK}
}

\setcopyright{none}

\usepackage{amsmath}
\allowdisplaybreaks
\usepackage{dsfont}

\usepackage{amssymb}

\usepackage{graphicx}
\usepackage{booktabs}
\usepackage{multirow}
\usepackage{algorithm}
\usepackage{algorithmic}

\newtheorem{proposition}{Proposition}
\newtheorem{remark}{Remark}

\newtheorem{example}{Example}

\begin{document}

\begin{abstract}
A novel form of inference attack in vertical federated learning (VFL) is proposed, where two parties collaborate in training a machine learning (ML) model. Logistic regression is considered for the VFL model. One party, referred to as the active party, possesses the ground truth labels of the samples in the training phase, while the other, referred to as the passive party, only shares a separate set of features corresponding to these samples.
It is shown that the active party can carry out inference attacks on both training and prediction phase samples by acquiring an ML model independently trained on the training samples available to them. 
This type of inference attack does not require the active party to be aware of the score of a specific sample, hence it is referred to as an agnostic inference attack. 
It is shown that utilizing the observed confidence scores during the prediction phase, before the time of the attack, can improve the performance of the active party's autonomous ML model, and thus improve the quality of the agnostic inference attack.
As a countermeasure, privacy-preserving schemes (PPSs) are proposed. 
While the proposed schemes preserve the utility of the VFL model, they systematically distort the VFL parameters corresponding to the passive party's features. 
The level of the distortion imposed on the passive party's parameters is adjustable, giving rise to a trade-off between privacy of the passive party and interpretability of the VFL outcomes by the active party. 
The distortion level of the passive party's parameters could be chosen carefully according to the privacy and interpretability concerns of the passive and active parties, respectively, with the hope of keeping both parties (partially) satisfied. 
Finally, experimental results demonstrate the effectiveness of the proposed attack and the PPSs. 
\end{abstract}

\maketitle


\section{Introduction}
The emergence of distributed machine learning (ML) techniques has been a game-changer in the way valuable information is obtained from raw data in various fields, such as computer vision, image recognition, financial services, and natural language processing. With the growing need for the utilization of distributed data to construct more accurate and advanced ML models, there has been a surge in demand in multiple sectors. To address the limitations of centralized ML models, including issues with data storage, excessive computations, security and privacy breaches, various distributed ML techniques have been proposed \cite{Wei_Fed_DP}. One such technique is federated learning (FL), which was introduced in \cite{McMahan_2017}. In FL, two or more data-centers, known as parties, work together to train a shared ML model, thereby addressing the concerns of traditional centralized learning. This approach has gained widespread attention for its applications in real-life scenarios, such as in health systems \cite{Songtao_health, Wenqi_health}, keyboard prediction \cite{Francoise_Keyboard, Andrew_keyboard_prediction}, and e-commerce \cite{Kai_2019, Wang_2020}.

Three types of Federated Learning (FL) are commonly studied in the literature, including Horizontal FL (HFL), Vertical FL (VFL), and Transfer FL. The distinction between these approaches lies in the manner in which data are partitioned and shared among parties. In HFL, each party holds a unique set of samples, but all parties share the same features \cite{Rong}. On the other hand, VFL involves parties sharing the same samples, but each party trains and predicts collaboratively using a separate set of features \cite{Cheng_2020}.
As shown in Figure \ref{fig122}, an illustration of VFL is presented, where a bank and a FinTech company collaborate in training a classifier that will be used by the bank to approve or reject credit card applications. The bank, as the active party, holds the ground truth labels for the samples in the training set. Meanwhile, the FinTech company, as the passive party, contributes to the VFL process by sharing a distinct set of features pertaining to the samples.

\begin{figure}[t]
 \centering 
 \scalebox{1} 
 {\includegraphics{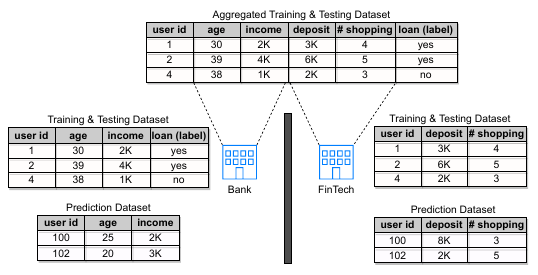}} 
 \caption{Digital banking as an example of vertical federated learning \cite{Xinjian}.}
 \label{fig122} 
\end{figure}

After the VFL model is trained, both the bank and the FinTech company use it collaboratively to make predictions for new samples. To do this, each party inputs their corresponding features for each sample into the model. Typically, a third party known as the coordinator authority (CA) is involved in this collaboration during both the training and prediction phases. The CA's role during the prediction phase is to obtain the model's outcome for new samples and share it with the active party when requested. The model output can be in the form of a confidence score, which is a probability vector representing the likelihood of each class, or a one-hot vector representing the output class label. Due to privacy concerns, the confidence score of a sample may be distorted before it is delivered to the active party, typically by providing a noisy or rounded version of the score. In the literature, generally two scenarios are considered for this type of collaboration, namely white-box and black box settings, where the active party is aware and unaware of the passive party's parameters, respectively.

Focusing on the prediction phase, in \cite{Xinjian}, it has been demonstrated that when the exact form of the confidence scores is shared with the active party, there is a risk of the adversary (i.e. the active party) reconstructing the features of the passive party, which are considered private information. This can result in a privacy leak for the users, and therefore known as inference attack. The reconstruction is accurate when the number of unknown features, or the passive party features, is less than the number of classes. In \cite{RaVaGu} various approaches regarding improved inference attacks have been proposed that are based on solving an under-determined system of linear equations. In particular, it is shown that using methods that approximate the best worst case solution among feasible solutions can significantly improve the attack accuracy. Furthermore, as a defense mechanism, it has been shown in \cite{RaVaGu} that by transforming the passive features\footnote{For brevity, throughout the paper we use passive (active) features and passive (active) parameters instead of passive (active) party's features and passive (active) party's parameters, respectively.} in a linear fashion and keeping the transformation mapping hidden from the active party, it is possible to preserve the model's accuracy while reducing the privacy leakage (i.e., higher mean squared error in reconstructing the passive features).

While previous works have studied inference attacks in the prediction phase using samples whose confidence scores (or a different version of them) are available at the active party, in this paper, under the white-box setting, we focus on the privacy leakage of samples whose confidence scores are not yet available at the active party. 
This can include samples in the prediction phase for which no score has been requested by the active party, as well as samples used in the training phase to train the VFL model. Since the adversary is unaware of the confidence score of the attack target, we refer to this type of attack as an \textit{agnostic inference attack}. Focusing on this type of privacy leakage, the main contributions of the paper are as follows. Focusing on this type of privacy leakage, the main contributions of the paper are as follows.

\begin{itemize}

    \item Building on the linear-algebraic feature reconstruction framework originally introduced in \cite{RaVaGu}, the active party uses its available resources, i.e., the training data with only the active features and labels, to reconstruct an independent classifier similar to the VFL model. This classifier is named the adversary model (AM). This is possible because the active party already has access to the ground truth labels of samples in the training set. 
    Unlike prior inference attacks that assume access to exact or perturbed confidence scores of the attack targets, the AM enables inference in an \emph{agnostic} setting where such scores are unavailable.
    
    The AM can be used for two purposes. First, as a means to estimate any sample's confidence score and conduct inference attacks on the passive features using the reconstruction methods introduced in \cite{RaVaGu}. The main aspect of this type of attack is that the adversary is agnostic to the confidence score or any other form of it corresponding to the attack target,  which puts not only prediction samples but also \emph{all training samples} at risk of privacy leakage.
    Second, by comparing the sample-by-sample scores obtained from the AM and the VFL model during the prediction phase, the active party can use the amount of score mismatch as a measure of the improvement in classification due to the VFL and decide whether to continue or stop requesting further collaborative predictions. This is particularly useful when the active party incurs a cost for each VFL model prediction made through the CA.

    \item Although the AM is primarily trained on the training samples, its ability to conduct an attack can be enhanced by taking into account the prediction samples and their associated confidence scores that have already been communicated to the active party. We consider scenarios where the confidence scores are delivered in their exact form to the active party\footnote{\label{fn:softscores}
We assume that confidence scores are delivered to the active party in their exact form (soft version).
This represents a worst-case scenario from a privacy perspective, as soft scores generally reveal more information than hard labels or coarsened outputs. 
At the same time, this assumption is practically motivated as in many real-world deployments, the active party requires calibrated confidence scores for downstream tasks such as ranking, risk assessment, threshold selection, or soft clustering (e.g., prioritizing decisions under budget or resource constraints).}. It is observed that incorporating these scores leads to significant improvements in the effectiveness of the attack. This refinement mechanism, referred to as the refined adversary model (RAM), aligns the AM’s outputs more closely with those of the jointly trained VFL model using only a small number of observed prediction scores. The RAM does not introduce a new reconstruction technique, but rather improves the quality of score approximation within the same underlying linear system used in prior inference attacks.
    
    This realization highlights the limitations of privacy-preserving schemes (PPSs) that operate solely at the score level, and motivates a shift in focus toward the model parameters themselves rather than the confidence scores or passive features, departing from strictly black-box or white-box views in existing VFL settings.

    \item We propose PPSs as a defense mechanism against agnostic inference attacks. The proposed schemes are designed to maintain the utility of the VFL model, i.e., delivering confidence scores in their exact form, while at the same time preventing adversaries from misusing the VFL collaboration. Rather than altering the confidence scores, the PPSs operate by systematically distorting the parameters corresponding to the passive party’s features, while keeping the underlying prediction model intact.
    
    The level of distortion applied to the passive parameters is adjustable and determines the degree of privacy leakage of the passive party, as measured by the MSE of reconstructed features. A higher level of distortion results in lower privacy compromise. This introduces an explicit privacy--interpretability trade-off, where interpretability at the active party is quantified via the deviation between the released and original passive parameters.
    
    To address this trade-off, we formulate PPS design as constrained optimization problems, including problems over Stiefel manifolds, that enable controlled parameter distortion while preserving model utility. These schemes generalize prior feature-transformation-based defenses by operating directly on model parameters and by explicitly accounting for agnostic inference attacks.
\end{itemize}

The paper is structured as follows. In Section \ref{SysMod}, an explanation of the system model under study is provided.
In Section \ref{Datasets}, the characteristics of the datasets utilized in this paper, as well as the model used for evaluating the results, are reviewed.
Section \ref{AgnIA} is dedicated to the introduction of agnostic inference attack and is divided into three sub-sections. In Section \ref{AgnIA_Prel}, the necessary background information required to perform an inference attack is discussed. Subsequently, in Section \ref{AgnIA_AM}, the procedure for obtaining and refining the AM is detailed, followed by a performance analysis of the agnostic inference attack in Section \ref{AgnIA_tech}.
The proposed PPSs against agnostic inference attack are described in detail in Section \ref{PPS}, which is further divided into two sub-sections.
In Section \ref{Privacy_prel}, the geometry of optimization algorithms with orthogonality constraints, specifically optimization over Stiefel manifolds, is reviewed to provide a better understanding of one of the proposed PPSs.
Additionally, an overview of the concept of interpretability in the context of the present work is also provided. In Section \ref{Privacy_algo}, the technical analysis of the proposed PPSs is presented.
In Section \ref{NR}, the experimental results are reported and discussed, along with a comprehensive evaluation of the performance of agnostic inference attack on real-world datasets under various scenarios and the performance of the proposed PPSs.
The paper concludes in Section \ref{conc}.

\textbf{Notations.} 
Matrices and vectors\footnote{All the vectors considered in this paper are column vectors.} are denoted by bold capital (e.g., $\mathbf{A,Q}$) and bold lower case letters (e.g., $\mathbf{b,z}$), respectively. 
Sets are denoted by capital letters in calligraphic font (e.g., $\mathcal{X},\mathcal{G}$) with the exception of the set of real numbers, which is denoted by $\mathbb{R}$.  The cardinality of a finite set $\mathcal{X}$ is denoted by $|\mathcal{X}|$. 
The transpose of $\mathbf{A}_{m\times k}$ is denoted by $\mathbf{A}^T$, and when $m=k$, its trace and determinant are denoted by $\textnormal{Tr}(\mathbf{A})$ and $\textnormal{det}(\mathbf{A})$, respectively. For an integer $n\geq 1$, the terms $\mathbf{I}_n$, $\mathbf{1}_n$, and $\mathbf{0}_n$ denote the $n$-by-$n$ identity matrix, the $n$-dimensional all-one, and all-zero column vectors, respectively, and whenever it is 
clear from the context, their subscripts are dropped. 
For integers $m\leq n$, we have the discrete interval $[m:n]\triangleq\{m, m+1,\ldots,n\}$, and the set $[1:n]$ is written in short as $[n]$. 
For $\mathbf{x}\in\mathbb{R}^n$ and $p\in[1,\infty]$, the $L^p$-norm is defined as $\|\mathbf{x}\|_p\triangleq(\sum_{i=1}^n|x_i|^p)^{\frac{1}{p}},p\in[1,\infty)$, and $\|\mathbf{x}\|_\infty\triangleq\max_{i\in[n]}|x_i|$.  
Throughout the paper, $\|\cdot\|$ (i.e., without subscript) refers to the $L^2$-norm and spectral norm for vectors and matrices, respectively. 
For a matrix $\mathbf{W}= [\mathbf{w}_1, \mathbf{w}_2, \ldots, \mathbf{w}_n]$, we have $\textnormal{Vec}(\mathbf{W}) = [\mathbf{w}_1^T, \mathbf{w}_2^T, \ldots, \mathbf{w}_n^T]^T$. The symbol $\otimes$ stands for the Kronecker product.

\section{System model}\label{SysMod}
\subsection{Machine learning (ML)}\label{SysMod_ML}
An ML model is a function $f_{\pmb{\theta}}:\mathcal{X}\to\mathcal{U}$ parameterized by the vector $\pmb{\theta}$, where $\mathcal{X}$ and $\mathcal{U}$ denote the input and output spaces, respectively. In this paper, we consider the supervised classification setting, where a labeled training dataset is used to train the model.

Assume that a training dataset $\mathcal{D}_{\textnormal{train}}\triangleq\{(\mathbf{x}_i,u_i)|i\in[n]\}$ is given, where each $\mathbf{x}_i$ is a $d_t$-dimensional example/sample and $u_i$ denotes its corresponding label. Learning refers to the process of obtaining the parameter vector $\pmb{\theta}$ in the minimization of a loss function, i.e.,
\begin{equation}
    \min_{\mathbf{\pmb{\theta}}}\frac{1}{n}\sum_{i=1}^n l(f_\mathbf{\pmb{\theta}}(\mathbf{x}_i),u_i)+\omega(\mathbf{\pmb{\theta}}),
\end{equation}
where $l(\cdot,\cdot)$ measures the prediction loss of  $f_{\pmb{\theta}}(\mathbf{x}_i)$, while the true label is $u_i, i\in[n]$. To avoid overfitting, a regularization term $\omega(\pmb{\theta})$ can be added to the optimization.

Once the model is trained, i.e., $\pmb{\theta}$ is obtained, it can be used for the prediction of any new sample. In practice, the prediction is a (probability) vector-valued, i.e., it is a vector of confidence scores as $\mathbf{c}=(c_1,c_2,\ldots,c_k)^T$ with $\sum_jc_j=1,c_j\geq 0,j\in[k]$, where $c_j$ denotes the probability that the sample belongs to class $j$, and $k$ denotes the number of classes. Classification can be done by choosing the class that has the highest probability.

\subsection{Vertical Federated Learning}\label{SysMod_VFL}
VFL is a type of distributed ML model training approach in which two or more parties are involved in the training process, such that they hold the same set of samples with disjoint sets of features. The main goal in VFL is to train a model in a privacy-preserving manner, i.e., to collaboratively train a model without each party having access to other parties' features. Typically, the training involves a trusted third party known as the CA, and it is commonly assumed that only one party has access to the label information in the training and testing datasets. This party is named \textit{active} and the remaining parties are called \textit{
passive}. Throughout this paper, we assume that only two parties are involved; one is active and the other is passive. The active party is assumed to be \textit{honest but curious}, i.e., it obeys the protocols exactly, but may try to infer passive features based on the information received. As a result, the active party is exchangeably referred to as the \textit{adversary} in this paper.

In the existing VFL frameworks, CA's main task is to coordinate the learning process once it has been initiated by the active party. 
During the training, CA receives the intermediate model updates from each party, and after a set of computations, backpropagates each party's gradient updates, separately and securely. 
To meet the privacy requirements of parties' datasets, cryptographic techniques such as secure multi-party computation (SMC) \cite{Andrew_SMC} or homomorphic encryption (HE) \cite{HE_IVAN} are used. 

Once the global model is trained, upon the request of the active party for a new record prediction, each party computes the results of their model using their own features. 
CA aggregates these results from all the parties, obtains the prediction vector (confidence scores), and delivers that to the active party for further action. 
Throughout this paper, we assume that the active party is not allowed to make any query about the scores of the samples in the training set. 
This is justified due to the fact that the active party already has the ground truth labels of the training set. 
Note that there could be some edge cases where confidence scores of all the samples (both in the training and prediction phases) are required, such as when the active party aims at forming a soft clustering among all the samples via Gaussian mixture model to determine a high-confidence region within each class\footnote{For instance, in the Bank-FinTech example, the bank may wish to grant credit approval for clients with low entropy confidence scores due to budget limitations.}. 
Where this is the case, i.e., exact or distorted versions of attack targets' scores are available at the active party, inference attack methods studied in \cite{RaVaGu} could be applied to reconstruct the passive features. 

As in \cite{Xinjian, RaVaGu}, we assume that the active party has no information whatsoever about the underlying distribution of passive features. However, it is assumed that the knowledge about the names, types and ranges of passive features is available to the active party to decide whether to participate in a VFL or not.

\begin{table}[t]
\caption{Details of the datasets} 
\centering 
\begin{tabular}{c c c c c} 
\hline\hline 
Dataset & \#Features & \#Classes & \#Records \\ [0.5ex] 
\hline 
Bank & 19 & 2 & 41188 \\ 
Adult & 13 & 2 & 48842 \\ 
Satellite & 36 & 6 & 6430\\
PenDigits & 16 & 10 & 10992 \\
Grid & 13 & 2 & 10000 \\[0.5ex] 
\hline 
\end{tabular}
\label{table_dataset} 
\end{table}

\section{Datasets and model}\label{Datasets}
\textbf{Datasets.} In this paper, we perform a comprehensive evaluation of our proposed methods on five real-world and publicly available datasets, namely Bank, Adult salaries, Pen-based handwriting digits, Electrical grid stability and Satellite. These datasets have been widely used in the literature and have been obtained from the Machine Learning Repository website \cite{MLR}. They have been selected to cover a wide range of classification tasks, including both binary and multi-class classification. The details of each of these datasets are provided in Table \ref{table_dataset}. Furthermore, to gain further insights into the characteristics of these datasets, we provide a visual representation of some of their empirical statistics in Figure \ref{fig0}.

In the Bank dataset, there are a total of 20 features. However, as stated in the dataset's description file on the Machine Learning Repository website \cite{MLR}, the 11th feature has a significant impact on the output target and should be disregarded in order to obtain a more accurate predictive model. As a result, the training process in our study is based on 19 features, as detailed in Table \ref{table_dataset}. Additionally, this dataset has 10 categorical features which need to be handled properly by the predictive model. Common ways to do this include one-hot encoding of the categorical features, mapping ordinal values to each category, mapping categorical values to their statistics, and so on. In this paper, we have chosen to map each categorical feature to its average within each class. This approach has been found to be effective for models such as Logistic Regression (LR) and Neural Networks (NN).

\begin{figure*}[ht]
 \centering 
 \scalebox{0.25} 
 {\includegraphics{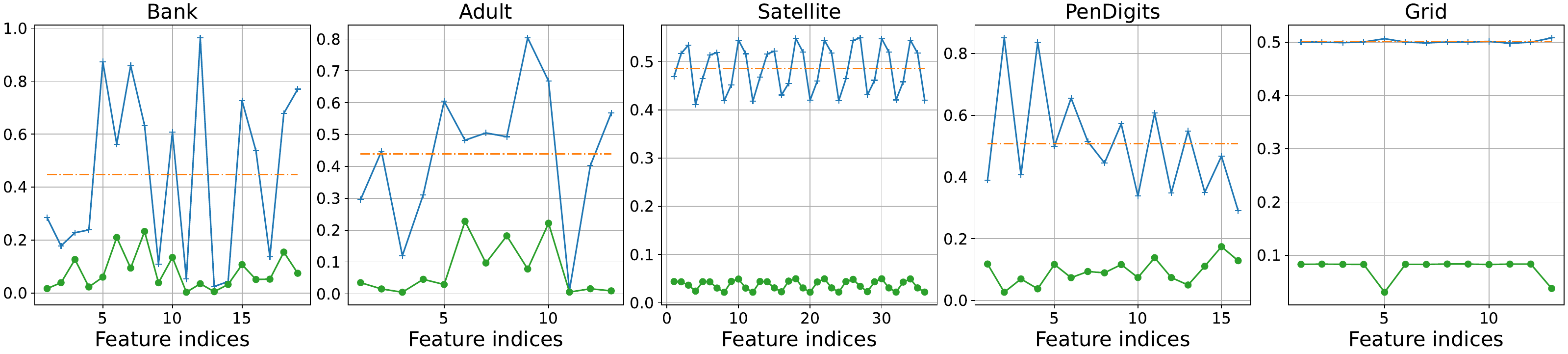}} 
 \caption{Empirical mean, variance, and mean of the means of features.}
 \label{fig0} 
\end{figure*}

\textbf{Model.} In this paper, we focus on LR, which can be modelled as 
\begin{align}\label{confi}
    \mathbf{c} = \sigma(\mathbf{Wx} + \mathbf{b}),
\end{align}
where $\mathbf{W}$ and $\mathbf{b}$ are the parameters collectively denoted as $\pmb{\theta}$, and $\sigma(\cdot)$ is the softmax function. 
Each party holds their parameters corresponding to their local features. The VFL model is trained in a centralized manner, which is a reasonable assumption according to \cite{Xinjian}. It is assumed that no intermediate information is revealed during the training phase, and only the final model is disclosed.

As in \cite{Xinjian}, the feature values in each dataset are normalized into $[0,1]$. It is important to note that normalizing the entire dataset, including both the training and test data, may lead to an overly optimistic model accuracy, a phenomenon known as \textit{data snooping} \cite{MacKinlay_Data_snooping}. This effect occurs when working with very noisy datasets and should be avoided in such cases. However, in our experiments, we have not considered this issue as the datasets under consideration do not seem excessively noisy.

Each dataset has been divided into 80\% training data and 20\% prediction data. The training set which is the one used for VFL model training is again split into 80\% training and 20\% test data. Data splitting is done using train$\_$test$\_$split in the scikit-learn package. In training LR, we apply early stopping, and the training is done without considering any regularization. ADAM optimization is used for training, and the codes, which are in PyTorch, are available online in our GitHub repository \cite{mrtzvrst}. 

\textbf{Baselines.}
Depending on the number of unknown features compared to the number of classes in a dataset, two different scenarios are considered. To approach these cases we use two different inference attack (estimation) methods, namely the least squares and half$^*$ methods. The latter is proposed in \cite{RaVaGu} which is also briefly explained in section \ref{AgnIA_Prel}.

\section{Agnostic Inference Attack}\label{AgnIA}

\subsection{Preliminaries}\label{AgnIA_Prel}
Let $(\mathbf{Y}^T,\mathbf{X}^T)^T$ denote random $d_t$-dimensional input samples, where the $(d_t-d)$-dimensional $\mathbf{Y}$, and the $d$-dimensional $\mathbf{X}$, correspond to the feature values held by the active and passive parties, respectively. 
We denote the samples in the training and prediction phase by the subscripts \textit{t} and \textit{p}, respectively, e.g., $\mathbf{X}_{p}$ represents the features of a sample in the prediction phase held by the passive party\footnote{When no distinction is required, we use the generic forms $\mathbf{Y}, \mathbf{X}$.}.

The VFL model under consideration is LR, where the confidence score is given by $\mathbf{c}=\sigma(\mathbf{z})$ with $\mathbf{z}=\bold{W}_{\textnormal{act}} \mathbf{Y}+\bold{W}_{\textnormal{pas}} \mathbf{X}+\mathbf{b}$. Denoting the number of classes in the classification task by $k$, $\bold{W}_{\textnormal{act}}$ (with dimension $k\times (d_t-d)$) and $\bold{W}_{\textnormal{pas}}$ (with dimension $k\times d$) are the model parameters of the active and passive parties, respectively, and $\mathbf{b}$ is the $k$-dimensional bias vector. From the definition of $\sigma(\cdot)$, we have 
\begin{align}\label{qe1}
    \ln \frac{c_{m+1}}{c_m} = z_{m+1}-z_m,\ m\in[k-1],
\end{align}
where $c_m,z_m$ denote the $m$-th component of $\mathbf{c},\mathbf{z}$, respectively. Define $\mathbf{J}$ as
\begin{equation}\label{JJ}
    \mathbf{J}\triangleq \begin{bmatrix}
    -1 & 1 & 0 & 0 & \ldots & 0\\
    0 & -1 & 1 & 0 & \ldots & 0\\
    0 & 0 & -1 & 1 & \ldots & 0\\
    \vdots & \vdots & \vdots & \vdots & \ddots &\vdots \\
    0 & \ldots & \ldots & 0 & -1 & 1
    \end{bmatrix}_{(k-1)\times k},
\end{equation}
whose rows are cyclic permutations of the first row with offset equal to the row index$-1$. By multiplying both sides of $\mathbf{z}=\bold{W}_{\textnormal{act}} \bold{Y}+\bold{W}_{\textnormal{pas}} \bold{X}+\bold{b}$ with $\mathbf{J}$, and using (\ref{qe1}), we get
\begin{align}
    \bold{JW}_{\textnormal{pas}}\bold{X} &= \mathbf{Jz}-\bold{J}\bold{W}_{\textnormal{act}}\mathbf{Y}-\mathbf{Jb}\label{eqeq1}\\
    &=\mathbf{c}'-\bold{J}\bold{W}_{\textnormal{act}}\mathbf{Y}-\mathbf{Jb},\label{eq:1}
    \end{align}
where $\mathbf{c}'$ is a $(k-1)$-dimensional vector whose $m$-th component is $\ln \frac{c_{m+1}}{c_m}$. Denoting the RHS of (\ref{eq:1}) by $\bold{b}'$, (\ref{eq:1}) writes in short as $\mathbf{Ax}=\mathbf{b}'$, where $\mathbf{A}\triangleq \bold{JW}_{\textnormal{pas}}$.

The white-box setting refers to the scenario where the adversary is aware of $(\mathbf{W}_{\textnormal{act}},\mathbf{W}_{\textnormal{pas}},\mathbf{b})$ and the black-box setting refers to the context in which the adversary is only aware of $\mathbf{W}_{\textnormal{act}}$.
In this paper, we are interested in a scenario where the active party wishes to reconstruct passive features of samples whose confidence scores $\mathbf{c}$ have not been observed yet  (which encompasses all the samples in the training phase as well). 
One measure by which the attack performance can be evaluated is the \textit{mean squared error} per feature, i.e.,
\begin{equation}\label{MSE}
    \textnormal{MSE}=\frac{1}{d}\mathds{E}\left[\|\mathbf{X}-\hat{\mathbf{X}}\|^2\right],
\end{equation}
where $\hat{\mathbf{X}}$ is the adversary's estimate. Let $N$ denote the number of samples whose passive features have been estimated. Assuming that these $N$ samples are produced in an i.i.d. manner, \textit{Law of Large Numbers} (LLN) allows to approximate MSE by its empirical value $\frac{1}{Nd}\sum_{i=1}^{N}\|\mathbf{X}_i-\hat{\mathbf{X}}_i\|^2$, since the latter converges almost surely to (\ref{MSE}) as $N$ grows.

\begin{remark}\label{rem:MMSE}
(On the use of MSE as a privacy metric): The MSE in (\ref{MSE}) is used as a measure of privacy leakage through feature reconstruction.
Although MSE is not a formal privacy definition, it admits a clear operational interpretation in the present setting: smaller MSE corresponds to more accurate reconstruction of the passive party’s features and hence higher inference risk. This choice can be theoretically motivated through information-theoretic connections.
In particular, the mutual information--minimum mean square error (I--MMSE) relationship~\cite{verdu} establishes that estimation distortion and information leakage are fundamentally linked.
Moreover, it is shown in \cite{Cuff} that differential privacy can be interpreted as a constraint on mutual information between private data and released outputs. 
Taken together, these results suggest that controlling reconstruction MSE implicitly limits information leakage, providing a meaningful proxy for privacy under the considered adversarial model\footnote{We emphasize that MSE does not provide worst-case or distribution-independent guarantees and does not capture all aspects of privacy.
Nevertheless, it is appropriate here because it aligns with prior work on inference attacks in VFL \cite{RaVaGu,Xinjian}, admits analytical characterization, and enables explicit optimization of the privacy--interpretability trade-off.}.
\end{remark}

\subsection{Adversary Model (AM)}\label{AgnIA_AM}


Recall that we study a scenario where prediction scores (exact forms or other) of attack targets are not available at the active party. 
When evaluating the potential inference attacks within this context, it is important to consider the passive party's vulnerabilities due to not just the white-box setting, but also the availability of data to the active party. 
This aspect of the scenario (the active party's data) can play a significant role in the success or failure of an attack. 
The active party's access to data and the quality of that data can greatly impact their ability to carry out an effective inference attack.
Despite the limitations imposed by the lack of access to prediction scores, the active party can still utilize the data and resources available to them to conduct an inference attack on passive features during both the training and prediction phases. To achieve this, the active party can build an independent AM using the training data at their disposal.
This allows the active party to estimate the confidence scores of their target, empowering them to carry out a successful attack.
To build an AM, the active party trains a classifier, denoted by $f_{\pmb{\theta}_a}:\mathcal{Y} \to \mathcal{U}$, using only the active features, independently of the passive party and CA. This is feasible as the active party already has the ground truth labels for their training set. The parameters of the AM are obtained through a training process applied over the active party's dataset, via the following minimization 
\begin{equation}\label{loss0}
    \min_{\pmb{\theta}_a}\frac{1}{n_t}\sum_{i=1}^{n_t} H(f_{\pmb{\theta}_a}(\mathbf{y}_{t, i}), u_i)+\omega(\pmb{\theta}_a),
\end{equation}
where $H(\cdot, \cdot)$ is the cross entropy loss function, $n_t$ is the number of training samples, and 
$\mathbf{y}_{t,i}$ is the active features of the $i$-th sample in the training set. $\pmb{\theta}_a$ denotes the parameters of AM. 

Training an independent classifier by the active party may appear to be a useful method to evaluate the improvement in performance due to collaboration with the passive party when building the VFL model. 
However, this approach can also be used for malicious purposes.
For example, the active party can use the function $f_{\pmb{\theta}_a}$ to estimate confidence scores, denoted by $\hat{\mathbf{c}}$, without even making a prediction query from CA.
These estimates can then be used to conduct inference attacks on passive features of any training or unseen sample.
This can be achieved by forming a new set of equations using the steps outlined in (\ref{qe1}) to (\ref{eq:1}), with $\mathbf{c}$ being replaced with $\hat{\mathbf{c}}$.
As discussed in further detail in section \ref{AgnIA_tech}, with this new set of equations, the adversary can use known methods for inference attacks \cite{RaVaGu} to estimate the passive features. Needless to say that this puts all samples in the training set and samples whose scores have yet to be observed at risk of privacy breaches. 

Postponing the technical analysis to section \ref{AgnIA_tech}, in Figure \ref{Bank_alone} the performance of agnostic inference attack (solid line) on the bank dataset is compared with two other baselines, namely the half (green solid$+$ line) and half$^*$ (blue solid$\bullet$ line) estimations.
In the former, the passive features are estimated as half (recall that features are normalized into $[0,1]$). In the latter, the half$^*$ estimation method, proposed in \cite{RaVaGu}, is used.
As it is observed from the results in Figure \ref{Bank_alone}, even in the absence of targets' confidence scores, the adversary manages to conduct an inference attack comparably well, especially for larger values of $d$ where the difference with the normal inference attack is negligible.

\begin{figure}[t]
 \centering 
 \scalebox{0.2} 
 {\includegraphics{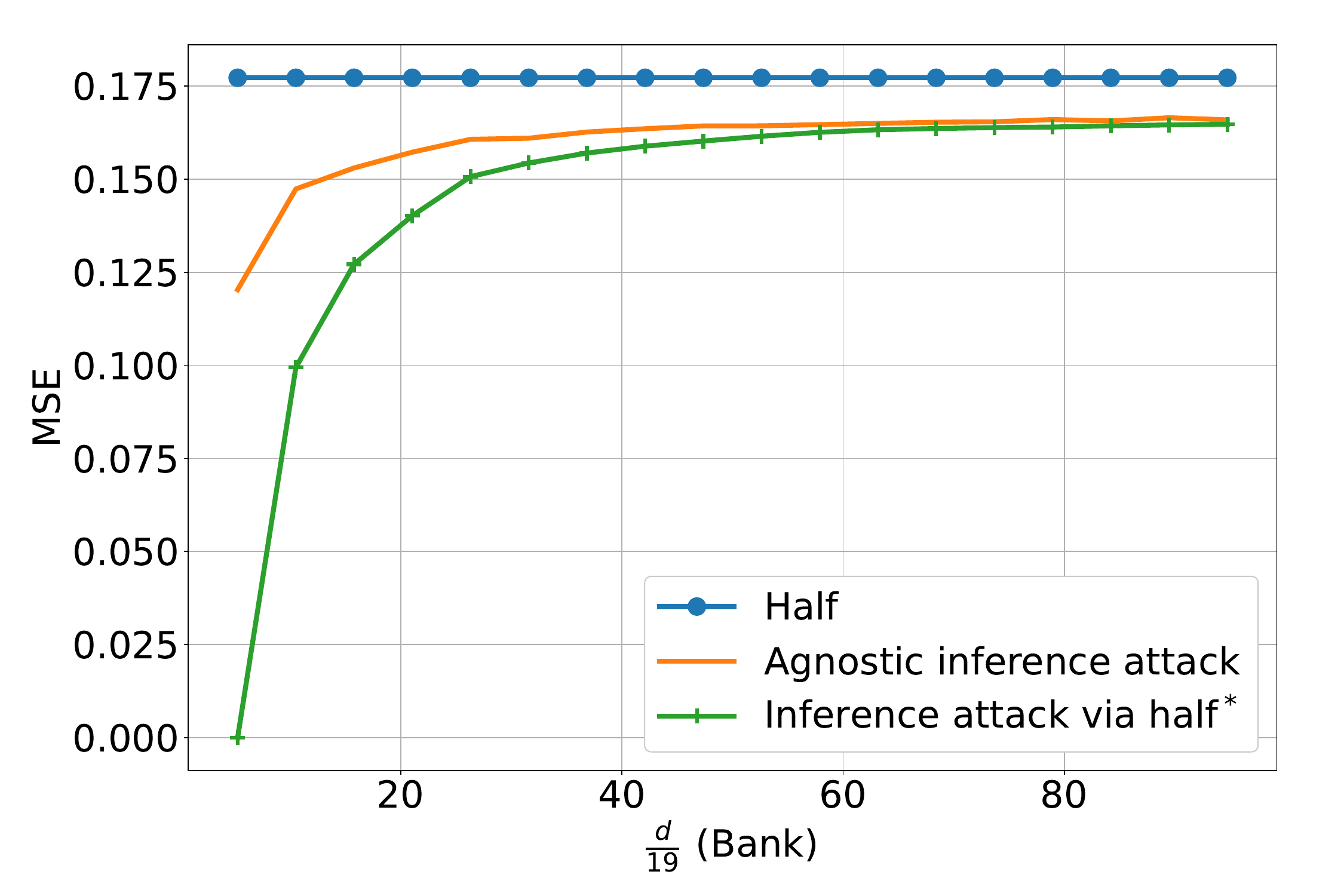}} 
 \caption{Comparison of MSE per feature obtained from half estimation, agnostic inference attack and half$^{*}$ method in \cite{RaVaGu}}
 \label{fig1} \label{Bank_alone}
\end{figure}

\subsubsection{Refined AM}
To further enhance the accuracy of the AM in estimating the targets' confidence scores, the active party can leverage the prediction scores received from CA before carrying out an inference attack on a specific target. These scores are generated by applying the VFL model that has been trained on the combination of both active and passive features.
As a result, they carry a higher level of confidence compared to the scores generated from $f_{\pmb{\theta}_a}$.
To take advantage of this increased reliability, the active party could adopt a training approach where the objective loss function incorporates the scores received from the CA. This can be done in a way that during the AM training, a higher cost is incurred in the event of a mismatch between the refined AM's scores and the scores received from CA.

The refined AM, denoted as $f_{\pmb{\theta}_a^r}$, is trained on the active party's available data by solving the following optimization problem
\begin{align}\label{loss1}
\min_{\pmb{\theta}_a^r} \biggl(&\frac{1}{n_t}\sum_{i=1}^{n_t} H(f_{\pmb{\theta}_a^r}(\mathbf{y}_{t, i}), u_{t,i}) + \frac{\alpha}{n_p}\sum_{j=1}^{n_p} S\left(f_{\pmb{\theta}_a^r}(\mathbf{y}_{p,j}), \mathbf{c}_j\right) + \omega(\pmb{\theta}_a^r)\biggr)
\end{align}

where $\pmb{\theta}_a^r$ denotes the parameters of the refined AM. $n_p$ denotes the total number of prediction samples (whose scores have already been delivered to the active party).
The hyperparameter $\alpha$ is used to stress the level of reliability of the prediction scores from CA as compared to the ones obtained by feeding the corresponding active features into AM.
Additionally, the scoring function, $S(\mathbf{c}^{'}, \mathbf{c})$, is defined as $S(\mathbf{c}^{'}, \mathbf{c}) =\sum_{i=1}^{k} {\log^2 (c^{'}_i/c_i)}$. 


\begin{remark}
The choice of the function $S(\cdot,\cdot)$ in equation (\ref{loss1}) is motivated by the desire to model the similarity (or soft mismatch) between the confidence scores delivered by CA and the ones produced by AM. 
The function $S(\mathbf{c}^{'}, \mathbf{c})$ is a quasiconvex\footnote{A function is quasi-convex if its sublevel sets are convex.} function that attains its minimum value of zero when $\mathbf{c}^{'}= \mathbf{c}$. By setting $\alpha > 0$, one can adjust the significance of samples during the prediction phase. This is essentially equivalent to enforcing the AM to overfit on the prediction confidence scores received from CA.
The value of the hyperparameter $\alpha$ in equation (\ref{loss1}) must be determined by the active party without knowledge of the actual values of the passive features. In our numerical results we consider $\alpha=1$.
\end{remark}

\subsection{Performance Analysis}\label{AgnIA_tech}
Upon agreement on a trained VFL model, during the prediction phase, the active party receives corresponding confidence scores (or a modified version of them, such as a noisy or rounded version, or the corresponding class label).
In the white-box setting, the active party can construct a system of linear equations as described in equation (\ref{eq:1}), where the unknown variables are the passive features.
If the number of unknowns, $d$, is equal to or less than the number of equations, $k-1$, and the confidence scores are provided in their exact form, the active party can recover the passive features by solving the linear system, $\mathbf{A}\mathbf{X}=\mathbf{b}'$.
However, if the confidence scores have been altered to satisfy privacy requirements, or if the number of unknowns, $d$, is greater than the number of equations, $k-1$, causing $\mathbf{A}$ to not have full column rank, there may exist either a different solution or an infinite number of solutions for the linear system, respectively.

A closer examination of (\ref{eq:1}) poses the following question: In a white-box setting, to what extent does the lack of confidence scores from the attack targets contribute to the average reconstruction error in (\ref{MSE})? Or equivalently, can an inference attack still be performed effectively in the absence of the actual confidence score?
In order to answer these questions, we evaluate the performance of an attack when an estimate of a confidence score obtained from AM is used instead of the actual value.
To that end, we consider two scenarios: i) $d<k$, and ii) $d\geq k$. 

\subsubsection{$d<k$}
In this case, we have an overdetermined system of linear equations, i.e., we have more equations ($k-1$) than unknowns ($d$). An overdetermined system can have solutions in some cases, e.g., when some equations are linear combinations of others. In the context of this paper, an example of an overdetermined system with a unique solution is when the active party has all the ingredients needed to solve (\ref{eq:1}). In this case, no matter which $d$ equations are chosen out of $k-1$, it yields the same solution. 
However, depending on how many of the equations are independent of each other, we can have other cases with no solution or infinite solutions. 
The case of interest where the active party does not have the attack target's confidence score, is an example of an overdetermined system with no solution.  
In this case, approximate solutions can be obtained via the method of ordinary least squares, i.e., $\min_{\mathbf{X}}\|\mathbf{AX}-\hat{\mathbf{b}'}\|$, the solution of which can be written with the normal equations as\footnote{Throughout the paper, $\hat{\mathbf{c}}'$ is adversary's estimation of $\mathbf{c}'$, and accordingly, $\hat{\mathbf{b}}'$ is an approximation of $\mathbf{b}'$ induced by $\hat{\mathbf{c}}'$.}
\begin{align}
\hat{\mathbf{X}}_{\textnormal{LS}} = (\mathbf{A}^T\mathbf{A})^{-1}\mathbf{A}^T\hat{\mathbf{b}}'.\label{M17}
\end{align}
With (\ref{M17}), an approximate solution is found when no exact solution exists ($\hat{\mathbf{b}}' \neq \mathbf{b}'$), and it gives an exact solution when one does exist (i.e., when $\hat{\mathbf{b}}' = \mathbf{b}'$).

Assuming that the adversary has access to an estimate of a confidence score, denoted by $\hat{\mathbf{c}}'$, different from the one in equation (\ref{eq:1}), an approximate solution denoted by $\hat{\mathbf{X}}_{\textnormal{LS},a}$ is obtained as below\footnote{Throughout the paper, subscript $a$ in estimations such as $\hat{\mathbf{X}}_{\textnormal{LS},a}$ is used to emphasize that the attack is agnostic.}
\begin{align}
\hat{\mathbf{X}}_{\textnormal{LS},a} &= (\mathbf{A}^T\mathbf{A})^{-1}\mathbf{A}^T(\hat{\mathbf{c}}'-\bold{J}\bold{W}_{\textnormal{act}}\mathbf{Y}-\mathbf{Jb})\\\label{M28}
&= (\mathbf{A}^T\mathbf{A})^{-1}\mathbf{A}^T(\mathbf{b}'+\hat{\mathbf{c}}'-\mathbf{c}').
\end{align}
The MSE of the estimation in (\ref{M28}) reads as (see Appendix \ref{app:1} for the derivation of (\ref{M18}))
\begin{equation}
    \textnormal{MSE}(\hat{\mathbf{X}}_{\textnormal{LS},a})=\frac{1}{d}\textnormal{Tr}\left(\mathbf{A}(\mathbf{A}^T\mathbf{A})^{-2}\mathbf{A}^T\mathbf{K}_{\mathbf{c}',\mathbf{\hat{c}}'}\right)\label{M18},
\end{equation}
where $\mathbf{K}_{\mathbf{c}',\mathbf{\hat{c}}'} \triangleq \mathds{E}[(\hat{\mathbf{c}}'-\mathbf{c}')(\hat{\mathbf{c}}'-\mathbf{c}')^T]$. 
From (\ref{M18}), it is easy to note that when the active party has the exact form of the scores, the MSE in this case is zero, i.e., perfect reconstruction. 
Additionally, note that (\ref{M18}) is non-negative due to the non-negativity of MSE.  

\subsubsection{$d\geq k$}
In \cite{RaVaGu}, several inference attack methods were evaluated, including RCC1, RCC2, CLS, and half$^*$. These methods were found to significantly improve the adversary's estimation quality, resulting in lower MSE compared to the methods discussed in \cite{Xinjian, Jiang}.
While the results and discussions for agnostic inference attack in this paper are general and apply to any of the feature estimation methods, for the purposes of this paper, we specifically focus on using the half$^*$ method.
This is because, on the one hand, the half$^*$ method is analytically feasible, making it easier to study. On the other hand, as indicated in \cite[Figure 4]{RaVaGu}, the half$^*$ reconstruction method was found to have the best performance, or at least performance comparable to the best method. Thus, before proceeding with our analysis, we briefly outline the half$^*$ estimation method.

Any solution of an underdetermined system of linear equations ($d\geq k$) can be written as $\mathbf{A}^{+}\mathbf{b}'+(\mathbf{I}_d-\mathbf{A}^+\mathbf{A})\mathbf{w}$ for some $\mathbf{w}\in\mathbb{R}^d$, where $\mathbf{A}^+$ denotes the pseudoinverse of $\mathbf{A}$ satisfying the Moore-Penrose conditions\cite{penrose}\footnote{When $\mathbf{A}$ has linearly independent rows, we have $\mathbf{A}^+=\mathbf{A}^T(\mathbf{AA}^T)^{-1}$.}. 
As elaborated further in \cite{RaVaGu}, when the only available information about $\mathbf{X}$ is that it belongs to $[0,1]^d$, then $\hat{\mathbf{X}}_\textnormal{half}=\frac{1}{2}\mathbf{1}_d$ is optimal in the best-worst sense.
Indeed, the naive estimate of $\frac{1}{2}\mathbf{1}_d$, is the Chebyshev center of $[0,1]^d$. The adversary can perform better when the side information $(\mathbf{b}',\mathbf{A})$ is available. The half$^*$ scheme is built on top of the naive $\frac{1}{2}\mathbf{1}_d$ estimation as follows. 
The estimator finds a solution in the solution space that is closest to $\frac{1}{2}\mathbf{1}_d$. In \cite[Proposition 1]{RaVaGu} it is shown that
\begin{equation}\label{half*}
   \hat{\mathbf{X}}_{\textnormal{half}^*}= \mathbf{A}^{+}\mathbf{b}'+\frac{1}{2}(\mathbf{I}_d-\mathbf{A}^+\mathbf{A})\mathbf{1}_d.
\end{equation}

As discussed in section \ref{AgnIA_AM}, although the confidence score is not available to the adversary, they manage to obtain an estimate of it via AM. Under these circumstances, to obtain an estimate of $\mathbf{c}'$, the adversary follows the lines similar to equations (\ref{qe1}), (\ref{JJ}), (\ref{eqeq1}) and (\ref{eq:1}). We have 
\begin{align}
    \bold{JW}_{a}\bold{Y} &= \mathbf{J}\mathbf{z}_a-\mathbf{J}\mathbf{b}_a\nonumber\nonumber\\
    &=\mathbf{\hat{c}}'-\mathbf{J}\mathbf{b}_a,
    \end{align}
where the subscript $a$ in $\mathbf{W}_a, \mathbf{b}_a, \mathbf{z}_a$ is to clarify the attribution of these parameters to AM. Noting that $\mathbf{\hat{c}}'$ is an estimate of $\mathbf{c}'$, the adversary conducts agnostic inference attack using half$^*$ method. Denoting this estimation by $\hat{\mathbf{X}}_{\textnormal{half}^*,a}$, we have
\begin{align}\label{M27}
    \hat{\mathbf{X}}_{\textnormal{half}^*,a} =& \mathbf{A}^+ (\mathbf{\hat{c}}'-\bold{J}\bold{W}_{\textnormal{act}}\mathbf{Y}-\mathbf{Jb})+ \frac{1}{2}(\mathbf{I}-\mathbf{A}^+\mathbf{A})\mathbf{1}_d\\
    =& \hat{\mathbf{X}}_{\textnormal{half}^*} + \mathbf{A}^+(\mathbf{\hat{c}}' - \mathbf{c}').\label{M35}
\end{align}
The MSE resulted from the estimation in (\ref{M35}) reads as (see Appendix \ref{app:2} for the derivation of (\ref{M2}))
\begin{align}
\textnormal{MSE}(\hat{\mathbf{X}}_{\textnormal{half}^*, a})= \frac{1}{d}\textnormal{Tr}\left(\left(\mathbf{I}-\mathbf{A}^+\mathbf{A}\right)\mathbf{K}_{\frac{1}{2}\mathbf{1}}\right)+ \frac{1}{d}\textnormal{Tr}\left({\mathbf{A}^+}^T\mathbf{A}^{+}\mathbf{K}_{\mathbf{c}',\mathbf{\hat{c}}'}\right)\label{M2}.
\end{align}
From (\ref{M2}) it is clear that the adversary's ability to carry out agnostic inference attack is dependent on the accuracy of the estimation of a target confidence score, as reflected through the term $\mathbf{K}_{\mathbf{c}',\mathbf{\hat{c}}'}$. 
The adversary's performance in this scenario cannot surpass the performance obtained when the exact form of confidence scores is known (see Appendix \ref{app:2}).
Note that, neither the active party nor the passive party have access to the exact form of the confidence scores or the parameters of AM, respectively, making it impossible for them to determine the exact value of (\ref{M2}). 
Despite this, we utilize (\ref{M2}) in section \ref{NR} to obtain our numerical results when evaluating the adversary's performance. 



It is worth mentioning here that the replacement of $\mathbf{c}'$ with $\hat{\mathbf{c}'}$ in (\ref{eq:1}) may result in some components of the final solution falling outside of the feasible region. 
To handle such cases, there are multiple strategies that can be implemented, including but not limited to: i) Mapping the individual components of the solution to the closest value of 0 or 1 if they are outside of the feasible region, ii) Mapping those individual components to the value of 0.5, or iii) Mapping all the passive features to 0.5 once an estimated feature of a sample falls outside the feasible region.
The best strategy depends on the underlying data statistics and requires deeper knowledge of the data which is indeed not available at the active party. Therefore, in order to ensure fairness in our numerical results, this paper adopts the first strategy in all datasets used in our experiments. 

\begin{remark}\label{rem:LR}
(Why LR and Extension to NN): We focus on LR as the primary model of interest for three reasons.
First, LR remains widely deployed in VFL applications where interpretability and auditability are essential, such as finance, healthcare, and risk assessment \cite{LR1,LR2,LR3, LR4}.
Second, LR admits a tractable analytical characterization, which enables explicit derivations of feature reconstruction error and principled optimization-based PPSs.
Third, softmax LR constitutes the final prediction layer in many NN architectures, making the notion of confidence scores and their approximation via adversary models directly relevant beyond pure LR. To elaborate further on this, in the following of this Remark, we consider a thought experiment on shallow NNs. Consider a shallow NN with two fully connected layers and $\tanh(\cdot)$ activations,
\[
c = \mathrm{softmax}\!\left(W_2 \tanh(W_1 \tilde{x} + b_1) + b_2\right),
\]
where $\tilde{x} = [y^\top, x^\top]^\top$ concatenates active and passive features.
In this setting, the active party can still train an AM or RAM using only local features and labels (and possibly observed scores) to approximate the confidence scores $c$.
However, unlike LR, the mapping from passive features $x$ to the logits becomes nonlinear. If the activation function is invertible (e.g., $\tanh$), and the passive contribution enters the network in a structurally identifiable manner, feature reconstruction could in principle be formulated as a nonlinear inverse problem.
In contrast, for non-invertible activations such as ReLU, the many-to-one mapping introduces intrinsic ambiguity, which can be interpreted as an additional source of distortion that may further degrade reconstruction accuracy.\footnote{For ReLU activations, the loss of injectivity implies that multiple passive feature values can yield identical activations, even when confidence scores are known.} A systematic treatment of these aspects is left for future work.
\end{remark}

\begin{remark}\label{rem:Col}
(Collusion and Auxiliary Information): Collusion between the active party and the CA or other parties can be modeled as an enlarged semi-honest adversarial coalition \cite{Wu2018PPGNN, Mohassel2017SecureML, Chowdhury2020Crypt}, where parties follow the protocol but share their views.
In this case, the active party may gain access to additional side information, such as partial passive features or additional confidence scores. Such auxiliary information can be incorporated into the RAM by augmenting the training objective with additional alignment or regularization terms, similar in spirit to auxiliary-data-based learning approaches \cite[see algorithm 2]{Jiang}.
Naturally, increased side information can improve the quality of confidence score approximation and, consequently, the effectiveness of inference attacks. We emphasize that adversarial models involving actively malicious behavior or protocol deviations by the CA fall outside the honest-but-curious threat model assumed in this paper and are left for future investigation.
\end{remark}

\begin{remark}\label{rem:Motiv}
(Motivation and implications of agnostic attacks):
In many real deployments—such as credit scoring (see Figure~\ref{fig122}), medical-risk prediction, or e-commerce recommendation—the passive features may include sensitive financial, behavioural, or clinical attributes that the passive party is unwilling to reveal directly. Using agnostic attacks, even when confidence scores corresponding to the target samples are not revealed, the active party can still reconstruct passive features using only its local data and model outputs. 
If the reconstructed features correlate strongly with sensitive attributes (e.g., income proxies in credit scoring or comorbidity indicators in medical diagnosis), this may lead to concrete privacy harms such as profiling, discrimination, or unintended leakage of protected characteristics. The effectiveness of agnostic attacks depends primarily on the statistical structure of the joint feature space. 
The risk is highest when (i) the active party holds a large portion of the predictive signal (labels and moderately informative features), and (ii) strong correlation exists between active and passive features. 
In such regimes, the AM and RAM can approximate score behaviour sufficiently well to support accurate linear reconstruction of passive features, even without access to the true confidence scores. This will be elaborated further in Section \ref{NR2}.
\end{remark}


\section{Privacy against agnostic inference attack}\label{PPS}
In the literature a number of algorithms have been proposed for privacy in the context of VFL in the white-box setting, such as those proposed in \cite{Xinjian, Jiang, RaVaGu}.
The majority of these algorithms have attempted to alter the confidence scores delivered to the active party in some way.
However, as demonstrated in Figure \ref{Bank_alone} and discussed in section \ref{AgnIA}, simply modifying the confidence scores is not enough to defend against agnostic inference attacks, where the adversary can build a separate classifier and improve their attack as they collect more scores from the CA. As such, it becomes clear that this approach is insufficient and leaves the passive party vulnerable to such attacks.

The central importance of the passive party's parameters in any potential adversarial attack highlights the need for effective defense techniques that specifically target these parameters.
A straightforward approach would be to implement a black box setting, in which the active party is not privy to the passive party's parameters, $\mathbf{W}_\textnormal{pas}$.
However, this approach sacrifices interpretability of the VFL model, which is a critical aspect in many applications where the client or user demands a clear explanation for the decisions made by the machine.
This creates a tension between the active party, who desires a white-box setting for interpretability purposes, and the passive party, who values privacy and prefers a black box setting.
Therefore, there is a need for a more nuanced and balanced approach to designing PPSs in the context of VFL, one that strikes a delicate balance between privacy and interpretability, and minimizes the cost of collaboration for all involved parties.  The following sections present an effort towards realizing this elusive goal of balancing privacy and interpretability in VFL.

The question we aim to address is: Can a PPS be designed that creates a balance between preserving the privacy of the passive party and enabling the interpretability capabilities of the active party?
To achieve this, we take into account the potential for agnostic inference attacks by an adversary using either (\ref{M28}) or (\ref{M27}) in different scenarios. Our proposed PPSs aim to systematically manipulate the parameters $\mathbf{W}_\textnormal{pas}$ in order to partially address the concerns of both parties and provide a mutually acceptable outcome. 

In the following, in subsection \ref{Privacy_prel}, we first review a particular case of optimization algorithms. Additionally, we provide an informal discussion on interpretability of parameters in LR. Then, in section \ref{Privacy_algo}, in a case-by-case fashion, we cover the technical analysis of the proposed PPSs. 

\begin{remark}\label{rem:Bla}
(Black-box Passive Parameters): The inference attacks analyzed in this paper rely on explicit access to the passive party’s parameters in order to construct linear systems for feature reconstruction.
In a black-box setting, where the active party does not know the passive parameters, this algebraic reconstruction mechanism is no longer directly applicable.
As a result, both reconstruction accuracy and interpretability at the active party are substantially reduced, while privacy for the passive party is strengthened. Although black-box access does not preclude other forms of attacks (e.g., model extraction or surrogate learning via adaptive queries), such attacks rely on fundamentally different assumptions and techniques.
Accordingly, this work focuses on the white-box regime, where the privacy--interpretability trade-off is most pronounced and analytically tractable.
\end{remark}

\subsection{Preliminaries}\label{Privacy_prel}
\subsubsection{Optimization problems with orthogonality constraint}\label{pr2}
Here, we provide a brief overview of an algorithm employed in this paper to determine the well-performing design parameters of one of the proposed PPSs. This algorithm is a specific instance of a larger class of optimization problems that incorporate an orthogonality constraint, specifically optimization problems over Stiefel manifolds.
While this particular case is reviewed in the sequel, for a more in-depth and comprehensive understanding of the geometry of the problem and the solver algorithms, we refer the readers to \cite{slpg3} and \cite{slpg1, slpg2}, respectively.

Consider the following optimization problem with an orthogonality constraint: 
\begin{equation}\label{M3}
\min_{\substack{\mathbf{R}:\\\mathbf{R}^T\mathbf{R}=\mathbf{I}_d}}\ \ f(\mathbf{R}),
\end{equation}
where $f: \mathbb{R}^{d\times d} \rightarrow \mathbb {R}$ satisfies the so-called \textit{blanket assumption}, i.e., the function $f(\mathbf{R})$ is differentiable and $\frac{\partial f}{\partial \mathbf{R}}$ is Lipschitz continuous\footnote{Assume two metric spaces $(\pmb{\mathcal{R}}, d_{\mathbf{R}})$ and $(\mathcal{X}, dx)$, where $d_{\mathbf{R}}$ and $dx$ denote the metrics on $\pmb{\mathcal{R}}$ and $\mathcal{X}$, respectively. A function $f: \pmb{\mathcal{R}} \rightarrow \mathcal{X}$ is called Lipschitz continuous if there exists a real constant $K \geq 0$, such that for all $\mathbf{R}_1$ and $\mathbf{R}_2$ in $\pmb{\mathcal{R}}$, we have \[d_x(f(\mathbf{R}_1),f(\mathbf{R}_2))\leq K\cdot d_{\mathbf{R}}(\mathbf{R}_1,\mathbf{R}_2)\].}. The orthogonality constraint in (\ref{M3}) can be expressed as $\mathcal{S}_d=\{\mathbf{R}\in \mathbb{R}^{d\times d}|\mathbf{R}^T\mathbf{R} = \mathbf{I}_d\}$, and is denoted as Stiefel manifold in real matrix space. 

\begin{remark}
The problem outlined in equation (\ref{M3}) plays a crucial role in various applications, such as ``discretized Kohn-Sham energy minimization'' and ``unsupervised feature selection.'' There are numerous methods available to address this problem, the majority of which utilize a trust-region strategy with quadratic approximation. 
In our numerical experiments, we utilize the PySTOP module based on the methods described in \cite{slpg1, slpg2}.  
\end{remark}

Almost all the algorithms tackling problems of the type (\ref{M3}), require gradient of the objective function $f(\mathbf{R})$ wrt the set $\mathcal{S}_d$. Denoting this gradient by $\nabla_s f(\mathbf{R})$, it is shown in \cite[eq. (2.53)]{slpg3}, that it is of the form 
\begin{equation}\label{M4}
\nabla_s f(\mathbf{R}) = \frac{\partial f}{\partial \mathbf{R}}-\mathbf{R}\left(\frac{\partial f}{\partial \mathbf{R}}\right)^T\mathbf{R}.
\end{equation}
One can view (\ref{M4}) as correcting the nominal gradient $\frac{\partial f}{\partial \mathbf{R}}$ by subtracting off its projection, i.e., $(\frac{\partial f}{\partial \mathbf{R}})^T\mathbf{R}$ onto the current solution $\mathbf{R}$. 
This ensures that the corrected gradient $\nabla_s f(\mathbf{R})$ is tangent to the manifold. 

In the following, we provide an analytically tractable example of this optimization (\ref{M3}) in the context of PPSs in VFL. As it will be elaborated later in section \ref{Privacy_algo}, providing analytical solutions for more interesting problems are still open if not infeasible. Therefore, in such cases we resort to iterative solver algorithms. 

\begin{example}
PPS for the least squares method: We aim at designing a PPS against an adversary who uses a naive least squares method. To that end, the passive parameters $\mathbf{W}_{\textnormal{pas}}$ are transformed via a secret orthonormal matrix $\mathbf{R}$ ($\mathbf{R}^T \mathbf{R} = \mathbf{I}$), and it is the transformed version $\mathbf{W}_{\textnormal{pas}} \mathbf{R}$ which is revealed to the active party. The adversary, having only the transformed parameters, forms the following system of equations $\mathbf{A}_n \mathbf{X} = \mathbf{b}'$ where $\mathbf{A}_n = \mathbf{JW}_{\textnormal{pas}} \mathbf{R} = \mathbf{AR}$\footnote{We assume that the adversary has been given or has estimated the exact form of the attack targets' confidence score.}.
Using least squares estimation method, the adversary obtains the estimated features as $\hat{\mathbf{X}}_{\textnormal{LS},a}^{\textnormal{PPS}} = \mathbf{A}_n^{+}\mathbf{b}'$\footnote{Superscript $\textnormal{PPS}$ is used to note that a privacy-preserving technique is in place at the time of an attack.}.
The MSE resulted from least squares inference attack when PPS is in place reads as (see Appendix \ref{app:3} for the derivation of (\ref{M26}))
\begin{equation}\label{M26}
   \textnormal{MSE}(\hat{\mathbf{X}}_\textnormal{LS}^\textnormal{PPS})=\textnormal{Tr}(\mathbf{I}+\mathbf{A}^{+}\mathbf{A})-2\textnormal{Tr}(\mathbf{R}\mathbf{K}_\mathbf{0}\mathbf{A}^{+}\mathbf{A}).
\end{equation}
To maximize $\textnormal{MSE}(\hat{\mathbf{X}}_{\textnormal{LS}}^{\textnormal{PPS}})$, and hence deteriorate the adversary's attack performance, we equivalently solve the following optimization problem:
\begin{equation}\label{M5}
\min_{\substack{\mathbf{R} \\ \mathbf{R}^T \mathbf{R} = \mathbf{I}_d}} f_{\textnormal{LS},a}^{\textnormal{PPS}}(\mathbf{R}) \triangleq \textnormal{Tr}(\mathbf{R} \mathbf{K}_0 \mathbf{A}^+ \mathbf{A}).
\end{equation}
Noting that the objective function in (\ref{M5}) is differentiable, and its derivative $\frac{\partial f_{\textnormal{LS},a}^{\textnormal{PPS}}}{\partial \mathbf{R}} = \mathbf{A}^+\mathbf{A}\mathbf{K}_{\mathbf{0}}$ is Lipschitz continuous\footnote{$\|\mathbf{A}^+\mathbf{A}\mathbf{K}_{\mathbf{0}}-\mathbf{A}^+\mathbf{A}\mathbf{K}_{\mathbf{0}}\|=0\leq K\cdot \|\mathbf{R}_1-\mathbf{R}_2\|$ for any $K\geq 0$ and any matrix norm.}, we have 
\begin{align}\label{M6}
   \nabla_s f_{\textnormal{LS},a}^{\textnormal{PPS}}&=\mathbf{A}^{+}\mathbf{A}\mathbf{K}_\mathbf{0} - \mathbf{R}\mathbf{K}_\mathbf{0}\mathbf{A}^{+}\mathbf{A}\mathbf{R}.
\end{align}

Let $ \mathbf{U}\pmb{\Sigma}\mathbf{V}^T$ be an SVD of $\mathbf{A}^{+}\mathbf{A}\mathbf{K}_\mathbf{0} $. It is easily verified that there are two candidates making the Stiefel gradient in (\ref{M6}) equal to zero, i.e., $\mathbf{R}=\pm \mathbf{UV}^T$. By substituting the candidate solutions in (\ref{M5}) we have 
\begin{equation}
   \textnormal{Tr}(\mathbf{R}\mathbf{K}_\mathbf{0}\mathbf{A}^+\mathbf{A})=\pm\textnormal{Tr}(\mathbf{UV}^T\mathbf{V}\pmb{\Sigma}\mathbf{U}^T)=\pm\textnormal{Tr}(\mathbf{U}\pmb{\Sigma}\mathbf{U}^T)=\pm\sum_{i=1}^{d}\sigma_i,
\end{equation}
where $\sigma_i$ is the i-th singular value of $\mathbf{A}^{+}\mathbf{A}\mathbf{K}_\mathbf{0}$. Noting that the singular values of a matrix are non-negative, the minimiser of (\ref{M5}) is $\mathbf{R} = -\mathbf{UV}^T$.
\end{example}

\begin{remark}
    A similar problem has been considered in \cite{RaVaGu}, where instead of transforming the parameters, the passive party transforms the features. In that study, the proof approach relies on Von Neumann's trace inequality, which can serve as an alternative to the optimization over Stiefel manifolds discussed in this paper. However, as will be made evident in the following sections, the optimization over Stiefel manifolds has a distinct advantage in that it allows for the treatment of a wide range of more complex problems through the use of iterative algorithms.
\end{remark}

\begin{remark}\label{rem1}
    It is worth noting that the example discussed in this section has a couple of limitations. Firstly, the assumption that the adversary will use the least squares method is unrealistic as there are other inference attack methods that are more advanced and efficient. Secondly, transforming the parameter (or feature space, as discussed in \cite{RaVaGu}), results in a significant and uncontrolled distortion of the parameters. This can make it difficult for the active party to interpret the results of the VFL model, and may even obstruct their ability to justify the outcomes in certain circumstances. For example, in the Bank-FinTech scenario, the bank may be required to provide a justification for rejecting a customer's credit card application. Thus, it is crucial that these limitations are taken into consideration when designing any PPS so as to ensure the active party's concerns are addressed.
\end{remark}

\subsubsection{Interpretability of the VFL model at the active party}\label{Privacy_Prel_Inte}
The increasing use of machine learning-based automation systems in sectors like government, banking, healthcare, and education presents both benefits and challenges. A significant challenge in VFL involves the interpretability of decisions made by models trained with PPSs, especially when these decisions impact legal, insurance, or health matters. A lack of outcome justifiability may erode trust between users and the system.

Interpretability is crucial in ML and pertains to understanding the relationship between a model's parameters and its outcomes. Simpler models such as linear and logistic regression, decision trees, and Lasso/Ridge regression are often favored over complex deep learning models for their interpretability when their performance is comparable. For instance, in logistic regression, the odds ratio between two classes for a sample can be calculated using $e^{(\mathbf{w}_i - \mathbf{w}_j) \mathbf{x}_1}$, where $\mathbf{w}_i$ is the $i$-th row of the parameter matrix $[\mathbf{W}_{\textnormal{act}}, \mathbf{W}_{\textnormal{pas}}]$. Consider a Bank-FinTech scenario where two customers, $\mathbf{x}_1$ and $\mathbf{x}_2$, have similar active features but differ in passive features; if one is approved and the other rejected for a credit application, the bank must justify the decision by the difference in passive features and the resultant change in the odds ratio, expressed as $e^{(w_{i,d_f} - w_{j,d_f})(x_{1,d_f} - x_{2,d_f})}$, where $d_f$ denotes the index of a specific passive feature.

\subsection{Privacy-preserving schemes (PPS)}\label{Privacy_algo}
When evaluating the effectiveness of agnostic inference attacks, we employ a ``genie-aided'' approach where the adversary is assumed to know the exact form of the confidence scores. This assumption is necessary because the performance of agnostic inference attacks as described in (\ref{M18}) and (\ref{M2}) heavily relies on the matrix $\mathbf{K}_{\mathbf{c}, \hat{\mathbf{c}}'}$. Given that the passive party lacks access to $\hat{\mathbf{c}}'$, precise calculations cannot be made without this assumption. By considering a genie-aided scenario, we recognize the adversary's potential capabilities to accurately estimate $\hat{\mathbf{c}}'$, whether through their own resources or via collusion with data owners as suggested by \cite{Xinjian}. Therefore, our objective is to mitigate the adversary's optimal attack performance under these conditions. It is advantageous for the passive party to develop PPSs that leverage the availability of $\mathbf{c}'$ at the active party to enhance privacy protection.

In designing the PPSs, our goal is to maintain the interpretability of passive parameters. To measure the distance between the accessible parameter set $\mathbf{W}_n$ and the original passive parameters $\mathbf{W}_{\textnormal{pas}}$, we define the function
\begin{equation}\label{M16}
g(\mathbf{R}) = \frac{1}{dk} \textnormal{Tr}\left((\mathbf{W}_{\textnormal{pas}}-\mathbf{W}_n)^T(\mathbf{W}_{\textnormal{pas}}-\mathbf{W}_n)\right),
\end{equation}
where $\mathbf{W}_n$ represents the parameters given to the active party instead of $\mathbf{W}_{\textnormal{pas}}$. \footnote{The function $g$ will take different arguments in the sequel depending on the mapping considered from $\mathbf{W}_{\textnormal{pas}}$ to $\mathbf{W}_n$.} To achieve our objective, we consider the following optimization problem:
\begin{equation}
\min_{\substack{\mathbf{R}:\\\mathbf{R}^T\mathbf{R}=\mathbf{I}_d}} f(\mathbf{R}) + \lambda \left(g(\mathbf{R}) - \varepsilon\right)^2,
\end{equation}
where $f(\mathbf{R})$ is a function such that its minimization corresponds to maximizing the overall MSE. Here, \(\lambda\) acts as a penalizing factor, imposing a certain level of distortion on the parameters while optimizing.

We study the PPSs in the following four distinct, disjoint cases:
\begin{itemize}
    \item[\textit{i})] For $d \geq k > 2$, we propose a heuristic approach based on transforming the passive parameters via orthonormal matrices.
    \item[\textit{ii})] For $1 < d < k$, a similar heuristic approach is applied.
    \item[\textit{iii})] For $d = 1$, we attempt to tackle the general problem without specific assumptions on the mapping from passive parameters to those accessible to the active party.
    \item[\textit{iv})] For $k = 2$ and $d > 1$, the approach again addresses the general problem without particular mapping assumptions.
\end{itemize}

\subsubsection{\textbf{Case} \pmb{$d\geq k>2$}}
In this case, the passive party's parameters $\mathbf{W}_\textnormal{pas}$ are transformed as \[\mathbf{W}_n = \mathbf{W}_\textnormal{pas}\mathbf{R},\] where $\mathbf{R}$ is assumed to be an orthonormal matrix, namely $\mathbf{R}^T\mathbf{R} = \mathbf{I}$\footnote{The assumption of orthonormality for the transforming matrix $\mathbf{R}$ ensures that our algorithm design remains analytically tractable. Without this assumption, solving the given problem becomes cumbersome. This complexity arises because, generally, for two matrices $\mathbf{A}_{k-1 \times d}$ and $\mathbf{R}_{d \times d}$, the equality $(\mathbf{AR})^{+} = \mathbf{R}^{-1}\mathbf{A}^{+}$ does not hold. There are some sufficient conditions under which this equality is valid, one of which is the orthonormality of $\mathbf{R}$ \cite{Aplus}. We leave the investigation of the general case of the transformer matrix $\mathbf{R}$ as an open problem for future research.\label{FN1}
}. 
Therefore, the function in (\ref{M16}) reads as
\begin{align}
    g(\mathbf{R}) & = \frac{1}{dk}\textnormal{Tr}(2\mathbf{W}-(\mathbf{R}+\mathbf{R}^T)\mathbf{W})\nonumber\\
    & = \frac{2}{dk}\textnormal{Tr}((\mathbf{I}-\mathbf{R})\mathbf{W})\label{M7},
\end{align}
where (\ref{M7}) is due to $\mathbf{R}^T\mathbf{R} = \mathbf{I}$, $\textnormal{Tr}(\mathbf{R}^T\mathbf{W}) = \textnormal{Tr}(\mathbf{RW})$ and $\mathbf{W}\triangleq \mathbf{W}_{\textnormal{pas}}^T\mathbf{W}_{\textnormal{pas}}$. 

Once the adversary receives the passive party's transformed parameters $\mathbf{W}_\textnormal{pas}\mathbf{R}$, they \textit{form} a system of linear equations $\mathbf{A}_n\mathbf{X}=\mathbf{b}'$\footnote{Here, we emphasize on the word ``form'' as this is passive party's deception of active party. Given $\mathbf{X}$, the true equation is $ \mathbf{AX}=\mathbf{b}'$.}, where $\mathbf{A}_n = \mathbf{JW}_\textnormal{pas}\mathbf{R} = \mathbf{AR}$. 
Note that $\mathbf{b}'$ is still the same as in (\ref{eq:1}). By applying half$^*$ method on this system, we have
\begin{align}
\hat{\mathbf{X}}_{\textnormal{half}^*,a}^{\textnormal{PPS}} &= \mathbf{A}^{+}_n\mathbf{b}' + \frac{1}{2}(\mathbf{I}-\mathbf{A}^{+}_n\mathbf{A}_n)\mathbf{1}_d\nonumber\\
&= \mathbf{R}^{T}\mathbf{A}^+\mathbf{b}' + \frac{1}{2}(\mathbf{I}-\mathbf{R}^{T}\mathbf{A}^{+}\mathbf{A}\mathbf{R})\mathbf{1}_d.
\end{align}
Therefore, the adversary's MSE corresponding to half$^*$ estimations reads as (see Appendix \ref{app:4} for the derivation of (\ref{eq1})) 
\begin{equation}\label{eq1}
    \textnormal{MSE}(\hat{\mathbf{X}}_{\textnormal{half}^*,a}^{\textnormal{PPS}})= \frac{1}{d}\textnormal{Tr}(\mathbf{K}_{\frac{1}{2}\mathbf{1}})-\frac{1}{d}\textnormal{Tr}\left(\mathbf{A}^{+}\mathbf{A}\left(2\mathbf{K}_\mathbf{0}\mathbf{R}^T-\mathbf{K}_\mathbf{0}-\mathbf{RMR}^T\right) \right).
\end{equation}

Our goal is to maximise the term in (\ref{eq1}) over orthonormal matrices $\mathbf{R}\in \mathcal{S}_d$, such that a certain level of interpretability is additionally preserved at the active party, i.e., $g(\mathbf{R})=\varepsilon$. 
To that end, we consider the following optimization problem
\begin{equation}\label{ob1}
\min_{\substack{\mathbf{R}:\mathbf{R}\in\mathcal{S}_d}}\ \ f_{\textnormal{half}^*,a}^{\textnormal{PPS}}(\mathbf{R})+\lambda (g(\mathbf{R})-\varepsilon)^2,
\end{equation}
where $f_{\textnormal{half}^*,a}^{\textnormal{PPS}}(\mathbf{R})$ denotes the second term in (\ref{eq1}). Note that $\varepsilon$ is a design parameter introduced into the model, potentially as part of a joint agreement between the active and passive parties. However, $\lambda \geq 0$ is an optimization tuning parameter, which should be considered large enough to enforce the equality $g(\mathbf{R})=\varepsilon$.

 
\begin{remark}\label{rem3}
The way the adversary obtains a system of linear equations in (\ref{eq:1}) is not unique; however, all such systems can be expressed in a form similar to (\ref{eq:1}) via $\mathbf{J}_{\textnormal{new}} = \mathbf{TJ}$, with $\mathbf{T}_{(k-1) \times (k-1)}$ being an invertible matrix, secretly used by the adversary. From this standpoint, it is crucial that, regardless of the choice of $\mathbf{T}$, the MSE in (\ref{eq1}), upon which the optimization in (\ref{ob1}) is based, remains consistent with the MSE experienced by the adversary. 
Therefore, any designed PPS should be independent of any transformation matrix $\mathbf{T}$ used by the adversary, as otherwise, intended privacy (certain MSE) and interpretability ($g(\mathbf{R}) = \varepsilon$) concerns are not guaranteed. From the objective function in (\ref{ob1}), the independence of $g(\mathbf{R})$ is obvious. For the function $f_{\textnormal{half}^*,a}^{\textnormal{PPS}}(\mathbf{R})$, note that its dependency on $\mathbf{J}$ is via $\mathbf{A}^+\mathbf{A}$.
Hence, if the adversary forms a new matrix $\mathbf{A}_{\textnormal{new}} = \mathbf{J}_{\textnormal{new}}\mathbf{W}_{\textnormal{pas}} = \mathbf{TJ}\mathbf{W}_{\textnormal{pas}}$, we have 
\begin{align}
    \mathbf{A}^{+}_\textnormal{new}\mathbf{A}_\textnormal{new}&=(\mathbf{J}_\textnormal{new}\mathbf{W}_\textnormal{pas})^+(\mathbf{J}_\textnormal{new}\mathbf{W}_\textnormal{pas})\nonumber=(\mathbf{TJ}\mathbf{W}_\textnormal{pas})^+(\mathbf{TJ}\mathbf{W}_\textnormal{pas})\nonumber\\
    &=(\mathbf{J}\mathbf{W}_\textnormal{pas})^+\mathbf{T}^{-1}\mathbf{T}(\mathbf{JW}_\textnormal{pas})=\mathbf{A}^{+}\mathbf{A}.\nonumber
\end{align}
Therefore, the solution of the objective function is independent of any orthogonal transformation of $\mathbf{J}$.
\end{remark}

If the blanket assumption is satisfied on the objective of the minimization problem in (\ref{ob1}), it fits within the category of the optimization problems reviewed in section \ref{Privacy_prel}. Continuity is obvious. The following proposition ascertains the Lipschitz continuity.
\begin{proposition}\label{prop2}
The objective function in (\ref{ob1}) is Lipschitz continuous.
\end{proposition}
\begin{proof}
The derivatives of $f_{\textnormal{half}^*,a}^{\textnormal{PPS}}(\mathbf{R})$ and $(g(\mathbf{R})-\varepsilon)^2$ are obtained as below\footnote{Recall that for matrices $\mathbf{A, B, C}$, we have  $\frac{\partial}{\partial \mathbf{A}}\textnormal{Tr}(\mathbf{AB}) = \mathbf{B}^T$ and $\frac{\partial}{\partial \mathbf{A}}\textnormal{Tr}(\mathbf{ABA}^T\mathbf{C}) = \mathbf{CAB}+\mathbf{C}^T\mathbf{AB}^T$}
\begin{align}
\frac{\partial f_{\textnormal{half}^*,a}^{\textnormal{PPS}}(\mathbf{R})}{\partial \mathbf{R}} &=\frac{1}{d}\left( -2\mathbf{A}^{+}\mathbf{A}\mathbf{K}_\mathbf{0}+\mathbf{A}^{+}\mathbf{A}\mathbf{R}(\mathbf{M}+\mathbf{M}^T)\right)\label{M11}\\
\frac{\partial}{\partial \mathbf{R}}  (g-\varepsilon)^2&= -\frac{4}{dk}\left(\frac{2}{dk}\textnormal{Tr}((\mathbf{I}-\mathbf{R})\mathbf{W})-\varepsilon\right)\mathbf{W}\label{M10},
\end{align}
where (\ref{M10}) is due to $\mathbf{W}=\mathbf{W}^T$. 
Lipschitz continuity of (\ref{M11}) is confirmed noting that
\begin{align}\nonumber
    \|\mathbf{A}^{+}\mathbf{A}\mathbf{R}_1(\mathbf{M}+\mathbf{M}^T)-\mathbf{A}^{+}\mathbf{A}\mathbf{R}_2(\mathbf{M}+\mathbf{M}^T)\|&= \|\mathbf{A}^{+}\mathbf{A}(\mathbf{R}_1-\mathbf{R}_2)(\mathbf{M}+\mathbf{M}^T)\|\\
    &\leq k_{\textnormal{half}^*}\|\mathbf{R}_1-\mathbf{R}_2\|\label{M12}
\end{align}
where in (\ref{M12}) we used the inequality $\|\mathbf{AB}\|\leq \|\mathbf{A}\| \|\mathbf{B}\|$ and we have $k_{\textnormal{half}^*}\triangleq\|\mathbf{A}^{+}\mathbf{A}\|\|(\mathbf{M}+\mathbf{M}^T)\|$ which is a constant by design. Note that (\ref{M12}) is valid for any matrix norm. Lipschitz continuity of (\ref{M10}) is also shown noting that  
\begin{align}\nonumber
    \|\textnormal{Tr}((\mathbf{R}_1-\mathbf{R}_2)\mathbf{W}+(\mathbf{R}_1-\mathbf{R}_2)^T\mathbf{W})\mathbf{W}\|&=2\|\textnormal{Tr}((\mathbf{R}_1-\mathbf{R}_2)\mathbf{W})\mathbf{W}\|\nonumber\\
    &\leq k_{g}\|\mathbf{R}_1-\mathbf{R}_2\|,\label{M13}
\end{align}
where $k_g\triangleq 2\sigma_1 \|\mathbf{W}\|$ with $\sigma_1$ being the largest singular value of the matrix $\mathbf{W}$. Additionally, in (\ref{M13}) we used trace-norm inequality\footnote{For two square matrices $\mathbf{A}, \mathbf{B}$, we have $\textnormal{Tr}(\mathbf{A}^T\mathbf{B})\leq \sigma_1\|\mathbf{B}\|$, which can be easily verified via Von Neumann's Trace Inequalities.}. Proof of the proposition follows by noting that the sum of two Lipschitz continuous functions is Lipschitz continuous. 
\end{proof}


\subsubsection{\textbf{Case} \pmb{$1<d< k$}}
In this case, the adversary uses (\ref{M17}) to solve the formed system of equations, i.e., $\hat{\mathbf{X}}_{\textnormal{LS},a}^{\textnormal{PPS}} = (\mathbf{A}_n^T\mathbf{A}_n)^{-1}\mathbf{A}_n^T \mathbf{b}'$. Recall that this is a case of an overdetermined system, where a unique solution is obtained when the adversary has access to all the parameters in (\ref{eqeq1}). Otherwise, an approximate solution is obtained. Following similar steps as in the previous case, and noting that the adversary uses $\mathbf{W}_n = \mathbf{W}_{\textnormal{pas}}\mathbf{R}$ for their attack, the resulting MSE reads as (see Appendix \ref{app:5} for the derivation of (\ref{M15}))
\begin{equation}\label{M15}
    \textnormal{MSE}(\hat{\mathbf{X}}_{\textnormal{LS},a}^{\textnormal{PPS}})= \frac{2}{d}\textnormal{Tr}\left(\mathbf{K}_\mathbf{0}-\mathbf{R}\mathbf{K}_\mathbf{0}\right).
\end{equation}
In (\ref{M15}), we use the identity $\textnormal{Tr}(\mathbf{R}^T\mathbf{K}_0) = \textnormal{Tr}(\mathbf{RK}_0)$. By calculating the Stiefel gradient of (\ref{M15}) and following the steps outlined in section \ref{pr2}, it is straightforward to verify that in the absence of interpretability constraints, the optimal transformation matrix is $\mathbf{R} = -\mathbf{I}$. This implies that the passive party simply changes the sign of their parameters. To also address the interpretability concerns of the active party, we form the following optimization problem: 
\begin{equation}\label{ob2}
\min_{\substack{\mathbf{R}:\mathbf{R}\in\mathcal{S}_d}}\ \ \frac{2}{d}\textnormal{Tr}(\mathbf{R}\mathbf{K}_0)+\lambda \left(g(\mathbf{R})-\varepsilon\right)^2.
\end{equation}
It is easily verified that the problem in (\ref{ob2}) satisfies the blanket assumption. 

\begin{remark}\label{rem4}
Note that when the active party has no interpretability concerns, it would be clear to them that the transformation matrix is $\mathbf{R} = -\mathbf{I}$. In such cases, they could easily recover the passive party's features by simply reversing the sign of the parameters. One potential solution to mitigate this issue is to maintain uncertainty for the active party about whether $\mathbf{R} = \mathbf{I}$ or $\mathbf{R} = -\mathbf{I}$ was used in the parameter mapping. Under this scenario, even if the adversary takes the average of the estimations from both cases, the convexity of the MSE would prevent them from achieving an estimation error lower than the average MSE of each case. Another approach to enhance privacy is to choose a random value for $\varepsilon$ and transform the parameters accordingly without disclosing the value of $\varepsilon$ to the active party.
\end{remark}

\subsubsection{\textbf{Case} \pmb{$d=1$}}
In this case, the passive party parameters are a $k$-dimensional vector, i.e., $\mathbf{w}_\textnormal{pas}$. As a result, the matrix $\mathbf{A}_{n}^T\mathbf{A}_n$ is a scalar. Consider $\mathbf{A}_n=\mathbf{Jw}_n$, where $\mathbf{w}_n$ is the vector the active party is given access to instead of $\mathbf{w}_\textnormal{pas}$. 
To conduct an attack, the adversary forms the equation $\hat{x}_{\textnormal{LS},a}^{\textnormal{PPS}} = (\mathbf{A}_{n}^T\mathbf{A}_n)^{-1}\mathbf{A}_n^T\mathbf{b}'$ resulting in the following MSE
\begin{align}
    \textnormal{MSE}(\hat{x}_{\textnormal{LS},a}^{\textnormal{PPS}}) &= \mathds{E}[x^2]\left(1-\frac{\mathbf{A}_n^T\mathbf{A}}{\mathbf{A}_{n}^T\mathbf{A}_n}\right)^2.\label{M29}
\end{align}
Therefore, the optimization problem of interest here is \footnote{note that $\mathds{E}[x^2]$ is removed in the optimization as it is a fixed positive parameter.}
\begin{align}
\min_{\substack{\mathbf{w}_n\in \mathbb{R}^{k\times 1}}}\ \ &\frac{\mathbf{A}_n^T\mathbf{A}}{\mathbf{A}_{n}^T\mathbf{A}_n}\left(2-\frac{\mathbf{A}_n^T\mathbf{A}}{\mathbf{A}_{n}^T\mathbf{A}_n}\right)\nonumber\\
&\textnormal{s.t.} \     \ g(\mathbf{w}_n)=\frac{1}{k}(\mathbf{w}_n-\mathbf{w}_\textnormal{pas})^T(\mathbf{w}_n-\mathbf{w}_\textnormal{pas})=\varepsilon\label{M32}
\end{align}

\begin{remark}\label{rem100}
For the special case of $d=1, k=2$, it can be easily verified that given the vectors $\mathbf{w}_n = [w_{1n}, w_{2n}]$ and $\mathbf{w}_\textnormal{pas} = [w_{1}, w_{2}]$, we have
\begin{align}\label{M33}
    \textnormal{MSE}(\hat{x}_{\textnormal{LS},a}^{\textnormal{PPS}}) &= \mathds{E}[x^2] \left(\frac{w_{n1} - w_1 + w_2 - w_{2n}}{w_{1n} - w_{2n}}\right)^2.
\end{align}
To maximise (\ref{M33}), given the equality constraint $\frac{1}{2}\left((w_1 - w_{n1})^2 + (w_2 - w_{n2})^2\right) = \varepsilon$, define $w_1 - w_{n1} \triangleq r \cos(\theta)$, $w_2 - w_{n2} \triangleq r \sin(\theta)$, and $d_0 \triangleq (w_1 - w_2)/r$. Rewriting the objective and the constraint in terms of $r, \theta$, we have
\begin{align}
\max_{\substack{r, \theta}}\ \ &\left(\frac{\sin(\theta) - \cos(\theta)}{d_0 + \sin(\theta) - \cos(\theta)}\right)^2.\nonumber\\
&\textnormal{s.t.} \     \ r^2 = 2\varepsilon\label{M19}
\end{align}
Given the equality constraint, i.e., $r = \sqrt{2\varepsilon}$, the optimization occurs on a circle with the only variable being $\theta$. By differentiating the objective in (\ref{M19}) with respect to $\theta$, it is verified that the necessary condition for the maximum value is $\theta \in \{\frac{3\pi}{4}, \frac{7\pi}{4}\}$. As a result, the maximizer of (\ref{M33}) can be either $[\mathbf{w}_\textnormal{pas} - \sqrt{\varepsilon}, \mathbf{w}_\textnormal{pas} + \sqrt{\varepsilon}]$ or $[\mathbf{w}_\textnormal{pas} + \sqrt{\varepsilon}, \mathbf{w}_\textnormal{pas} - \sqrt{\varepsilon}]$. The designer can try both configurations and select the one which maximizes (\ref{M33}) the most. From section \ref{Privacy_Prel_Inte}, observe that interpretability is not compromised if both parameters shift in the same direction. This is also easy to verify from (\ref{M33}), where such a case results in zero MSE. Note that the same approach in this remark can be used to convert the problem in (\ref{M32}) into an optimization over an n-Sphere with a fixed radius (due to the equality constraint), thus eliminating one variable and converting the problem into an unconstrained optimization over the surface of an n-Sphere, which could accelerate the optimization process.
\end{remark}

\begin{remark}
For larger values of $\varepsilon$, there can be possibly an infinite number of solutions for (\ref{M32}) that result in $\mathbf{A}_n^T\mathbf{A}_n = \mathbf{w}_n^T\mathbf{J}^T\mathbf{J}\mathbf{w}_n = 0$. This could occur for vectors with equal components that are within the $\varepsilon$-vicinity of $\mathbf{w}_\textnormal{pas}$. Considering the equality constraint significantly affects the demanding task of discarding such solutions. Therefore, one possible direction to speed up the computation in (\ref{M32}) is to replace the equality constraint with an inequality,
\end{remark}

\begin{remark}
An interesting observation here is that, in the case of $d=1$, the PPS does not depend on the statistics of the data. This undermines the idea of preserving the passive party's privacy as the active party can easily replicate the same steps to solve (\ref{M32}). To overcome this situation, the guidelines provided in Remark \ref{rem4} could be applied. Specifically, the passive party may either remain non-transparent about whether a PPS is in effect, or they could choose a random and secret $0 \leq \varepsilon' \leq \varepsilon$.
\end{remark}

\begin{remark}
In this scenario, any transformation of the parameters similar to cases \textit{i} and \textit{ii} (here we have $\mathbf{R} \in \mathbb{R}$ is a scalar), would result in scaling all the parameters by a constant factor, meaning $\mathbf{w}_n = r \mathbf{w}_\textnormal{pas}$. This ensures that the system of equations established by the adversary remains as an overdetermined system with a unique solution. Consequently, no matter which component of $\mathbf{w}_n$ the adversary selects, the attack equation would take the form $w_{n,i}x = rw_\textnormal{pas,i}x = b_i'$, where $w_{n,i}$, $w_\textnormal{pas, i}$, and $b_i'$ are the $i$-th components of $\mathbf{w}_n$, $\mathbf{w}_\textnormal{pas}$, and $\mathbf{b}'$, respectively. It is evident that, in this case, not only a fixed scaling but any linear transformation of the parameters (which results in an overdetermined system with no solution) cannot keep the parameters secret indefinitely. As discussed in \cite{RaVaGu}, the active party can collect valuable statistical information from received confidence scores (if the exact forms are provided) over time, which could be used to estimate the actual values of the parameters.
\end{remark}

\subsubsection{\textbf{Case} \pmb{$k=2, d>1$}}
In this case, the passive party's parameters are in the form of a matrix with dimensions $(2 \times d)$. Since $d \geq k$, the adversary is dealing with an underdetermined system, and therefore, we assume half$^*$ estimation is employed by the adversary. Noting that the pseudo-inverse of a vector is given by its transpose divided by its squared norm\footnote{The pseudo-inverse of an all-zero vector is an all-zero vector transposed.}, and assuming $\mathbf{A}_n = \mathbf{Jw}_n$, the resulting MSE is given as\footnote{As the steps to obtain (\ref{M20}) are similar to the steps leading to (\ref{eq1}), they have been omitted for brevity.}
\begin{equation}\label{M20}
    \textnormal{MSE}(\hat{\mathbf{X}}_{\textnormal{half}^*,a}^{\textnormal{PPS}})= \frac{1}{d}\textnormal{Tr}(\mathbf{K}_{\frac{1}{2}\mathbf{1}})+\frac{1}{d\|\mathbf{A}_n\|^2}(\mathbf{AK}_\mathbf{0}\mathbf{A}^T-2\mathbf{AK}_\mathbf{0}\mathbf{A}_n^T+\mathbf{A}_n\mathbf{MA}_n^T).
\end{equation}
(\ref{M20}) is independent of any transformed $\mathbf{J}$ used by the adversary. This is because 
\begin{equation}\nonumber
\mathbf{AK}_\mathbf{0}\mathbf{A}^T-2\mathbf{AK}_\mathbf{0}\mathbf{A}_n^T+\mathbf{A}_n\mathbf{MA}_n^T = \textnormal{Tr}(\mathbf{A}^T\mathbf{AK}_\mathbf{0}-2\mathbf{A}_n^T\mathbf{AK}_\mathbf{0}+\mathbf{A}_n^T\mathbf{A}_n\mathbf{M}).
\end{equation}
To find a right parameter set $\mathbf{w}_n$ to be given to the active party, we solve the following optimization problem 
\begin{align}
\min_{\mathbf{w}_n}\ \ &\frac{1}{d\|\mathbf{A}_n\|^2}(2\mathbf{AK}_\mathbf{0}\mathbf{A}_n^T-\mathbf{AK}_\mathbf{0}\mathbf{A}^T-\mathbf{A}_n\mathbf{MA}_n^T).\nonumber\\
&\textnormal{s.t.} \    \ g(\mathbf{w}_n)=\varepsilon\label{M34}
\end{align}
\begin{remark}
The optimization problems in (\ref{M32}) and (\ref{M34}) for $d=1$ and $k=2$, respectively, are more general compared to the transforming matrices proposed in cases i and ii.
However, since the implementation of the pseudo-inverse of a matrix is technically challenging, it restricts the extension of these cases to cases i and ii. 
\end{remark}

\begin{remark}
In (\ref{M20}), two special cases are worth reviewing. First, if $\mathbf{w}_n=\mathbf{w}_\textnormal{pas}$, we have
\begin{align}
\textnormal{MSE}(\hat{\mathbf{X}}_{\textnormal{half}^*,a}^{\textnormal{PPS}})&= \frac{1}{d}\textnormal{Tr}
(\mathbf{K}_{\frac{1}{2}\mathbf{1}}-\frac{\mathbf{A}^T\mathbf{A}}{\|\mathbf{A}_n\|^2}(\mathbf{K}_\mathbf{0}-\mathbf{M}))\nonumber\\
&= \frac{1}{d}\textnormal{Tr}
(\mathbf{K}_{\frac{1}{2}\mathbf{1}}-\mathbf{A}^{+}\mathbf{A}\mathbf{K}_{\frac{1}{2}\mathbf{1}}),
\end{align}
which is the same as the first term in $\ref{M9}$ (MSE obtained from applying half$^*$ using the exact form of $\mathbf{c}'$). Take any $d$-dimensional vector, namely $\mathbf{w}'_1$ with non-zero components, and form $\mathbf{w}'_n = [\mathbf{w}'_1, \mathbf{w}'_1]^T$. Set $\varepsilon = g(\mathbf{w}'_n)$. In this case, while the constraint in (\ref{M34}) is satisfied, the objective in (\ref{M34}), $\mathbf{A}'_n = \mathbf{Jw}'_n = \mathbf{0}^T$, is not defined. Additionally, any point in the vicinity of $\mathbf{w}'_n$ holds zero value in terms of interpretability of the VFL model. Noting that for large enough $\varepsilon$, there is an infinite number of such undefined points, one should carefully discard them when solving (\ref{M34}) numerically or enforce the algorithm to avoid such solution structures, which could be a demanding task in some cases. Alternatively, a relaxed version of the optimization problem in (\ref{M34}) could be considered as follows:
\begin{align}\label{M36}
\min_{\mathbf{w}_n}\ \ &\frac{1}{d}(2\mathbf{AK}_\mathbf{0}\mathbf{A}_n^T-\mathbf{AK}_\mathbf{0}\mathbf{A}^T-\mathbf{A}_n\mathbf{MA}_n^T).\nonumber\\
&\textnormal{s.t.} \ \ \begin{array}{cc}
     \|\mathbf{A}_n\|^2 \geq \varepsilon_1  \\
     g(\mathbf{w}_n)=\varepsilon.
\end{array}
\end{align}
Applying the Lagrange multiplier's method, the necessary condition for a stationary point in (\ref{M36}) is obtained as below
\begin{equation}\label{M37}
    \frac{1}{d}(2\mathbf{J}^T\mathbf{AK}_0-\lambda\mathbf{w}_\textnormal{pas}) = \mathbf{J}^T\mathbf{J}\mathbf{w}_n(\mathbf{M}+\mathbf{M}^T)-2\lambda_1\mathbf{J}^T\mathbf{Jw}_n,
\end{equation}
where $\lambda$ and $\lambda_1$ are Lagrange multipliers of $\varepsilon$ and $\varepsilon_1$, respectively. The solution to (\ref{M37}) can be found \cite{Matrixcookbook} to be
\begin{equation}\label{M38}
    \textnormal{Vec}(\mathbf{w}_n)= \frac{1}{d}((\mathbf{M}+\mathbf{M}^T-2\lambda_1\mathbf{I})\otimes \mathbf{J}^T\mathbf{J})^{-1}\cdot\textnormal{Vec}(2\mathbf{J}^T\mathbf{AK}_0-\lambda\mathbf{w}_\textnormal{pas}).
\end{equation}
\end{remark}

Given that the inverse in (\ref{M38}) exists, a lookup table could be created according to different values of $\lambda, \lambda_1$ (corresponding to different values of $\varepsilon, \varepsilon_1$) to be used according to the privacy interpretability concerns of the parties.

\begin{remark}\label{rem:JT}
(Robustness of PPSs and potential attack):
In the proposed PPSs, the passive party applies a transformation to its parameters that is determined by a secret design variable $R$.
We explicitly assume that this transformation is \emph{not} disclosed to the active party and that no auxiliary information (e.g., partial passive features obtained via collusion or other side channels) is available. Note that the auxiliary matrix $J$ used in constructing the linear reconstruction systems is fixed by definition.
Furthermore, as long as $R$ remains secret, different choices of matrix $T$ yield identical reconstruction performance in terms of MSE, and therefore do not reduce the privacy guaranties enforced by the PPS.
Consequently, although the active party may adapt the architecture of its AM or RAM models, such adaptations do not overcome the privacy protection provided by the PPS.
The reconstruction error is fundamentally governed by the distortion introduced through the secret transformation $R$.
Note that an active party may continuously collect prediction outputs over time and attempt to jointly estimate both model parameters and passive features by leveraging a large number of observations and imposing feasibility constraints on the feature domain.
Intuitively, such an approach would require a substantial number of prediction queries to achieve a meaningful estimation accuracy, with the required number growing with both the model dimension and the number of passive parameters.
A systematic analysis of the sample complexity required for joint parameter and feature estimation under these conditions is beyond the scope of this paper and is left for future work.
\end{remark}

\begin{remark}\label{rem:Adv}
(Adversarial Choice of PPS Parameters): Under the honest-but-curious assumption, the passive party should anticipate that the active party may construct AM or RAM models to assess confidence scores or decide whether to continue collaboration.
Accordingly, the passive party’s role is to select the PPS parameters to balance privacy protection and interpretability. In particular, the distortion budget $\varepsilon$ governing parameter perturbation is application-dependent and cannot be universally prescribed.
Instead, $\varepsilon$ may be calibrated by enforcing stability on downstream decisions or explanations (e.g., ensuring that model outputs or rankings remain within acceptable tolerances on a validation set), while maximizing the reconstruction error of passive features.
This calibration philosophy parallels the deployment of differential privacy mechanisms, where privacy parameters are typically chosen based on policy, regulatory constraints, and empirical utility considerations rather than closed-form optimality.
\end{remark}

\begin{remark}\label{rem:intpri}(Interpretability–privacy trade-off):
As discussed in Section \ref{Privacy_Prel_Inte}, revealing the passive-party parameters allows the active party to interpret how its local features interact with the passive features in the learned model. 
When the PPS introduces distortion, the resulting parameters deviate from their true values, reducing the fidelity of such interpretations. 
Operationally, higher distortion improves privacy by increasing reconstruction MSE, but decreases interpretability by weakening the correspondence between the released parameters and the true influence of passive features.
The PPS therefore enables a tunable trade-off: low distortion preserves interpretability but yields weaker privacy protection, whereas higher distortion increases privacy at the cost of interpretability.
This trade-off can be adjusted according to the requirements of the deployed VFL application.
\end{remark}

\section{Experimental Results}\label{NR}
In this section, we evaluate the performance of adversary's inference attack on passive features according to the different setups proposed in the previous sections. 
The performance of inference attacks and PPSs are evaluated according to the MSE per feature in (\ref{MSE}), which can be estimated empirically by $\frac{1}{Nd}\sum_{i=1}^N\|\mathbf{X}_i-\hat{\mathbf{X}}_i\|^2$ with $N$ denoting the number of samples under study and $d$ denoting the number of passive features. Since agnostic inference attack can be applied on samples both in the training and prediction phase, we set $N=n_t+n_p$ for each dataset, and name this averaging over $N$ as \textit{average over time}. This is to distinguish from another type of averaging, namely, \textit{average over space}, which is explained via an example as follows. Assume that the Bank data, which has 19 features, is considered. Also, consider the case that we are interested in obtaining the MSE when the active and passive parties have 14 and 5 features, respectively. Since these 19 features are not i.i.d., the MSE depends on which 5 (out of 19) features are allocated to the passive party. In order to resolve this issue, we average the MSE over some different possibilities of allocating 5 features to the passive party. More specifically, we average the MSE over a moving window of size 5 features, i.e., MSE is obtained for 19 scenarios where the feature indices of the passive party are $[1:5], [2:6],[3:7],\ldots,\{19\}\cup[1:4]$. Afterwards, these 19 MSE's are summed and divided by 19, which denotes the MSE when $d=5$. 

In the following, in section \ref{NR1}, we first review the experimental results related to the agnostic inference attack studied in this paper. Then, in section \ref{NR2}, we focus on the performance of the proposed PPSs against this type of attack. Finally, in section \ref{subsec:overhead}, we discuss about the computational and system-level overhead introduced by the PPSs.
\begin{figure*}[t]
 \centering 
 \scalebox{0.19} 
{\includegraphics{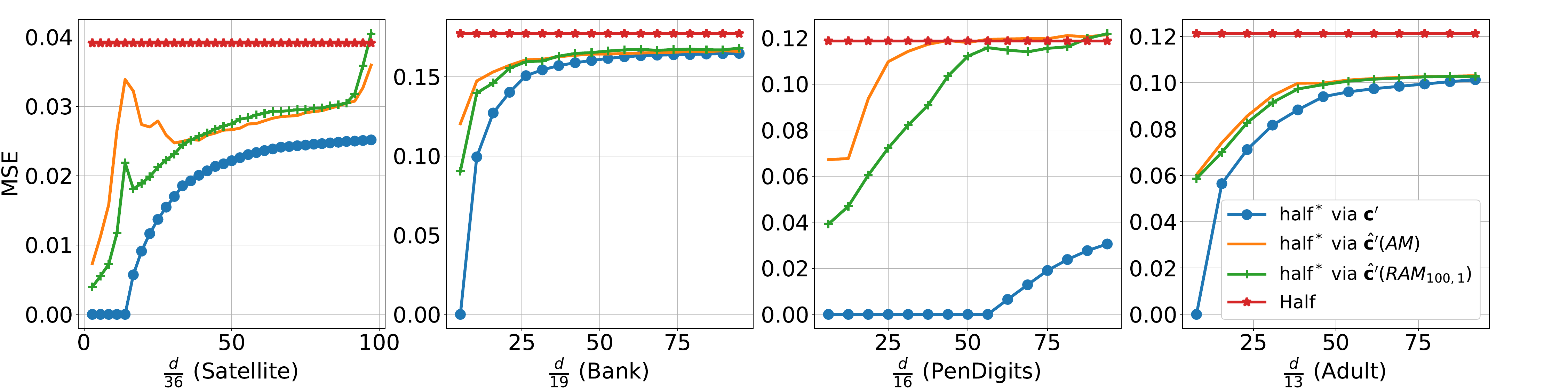}} 
 \caption{MSE per feature obtained from agnostic inference attack via $\hat{\mathbf{c}}'$ obtained from AM, $\hat{\mathbf{c}}'$ obtained from refined AM (RAM), half$^{*}$ via the exact form of $\mathbf{c}'$ obtained from the VFL model and half estimation.}
 \label{fig2} 
\end{figure*}

\begin{figure}[t]
 \centering 
 \scalebox{0.2} 
 {\includegraphics{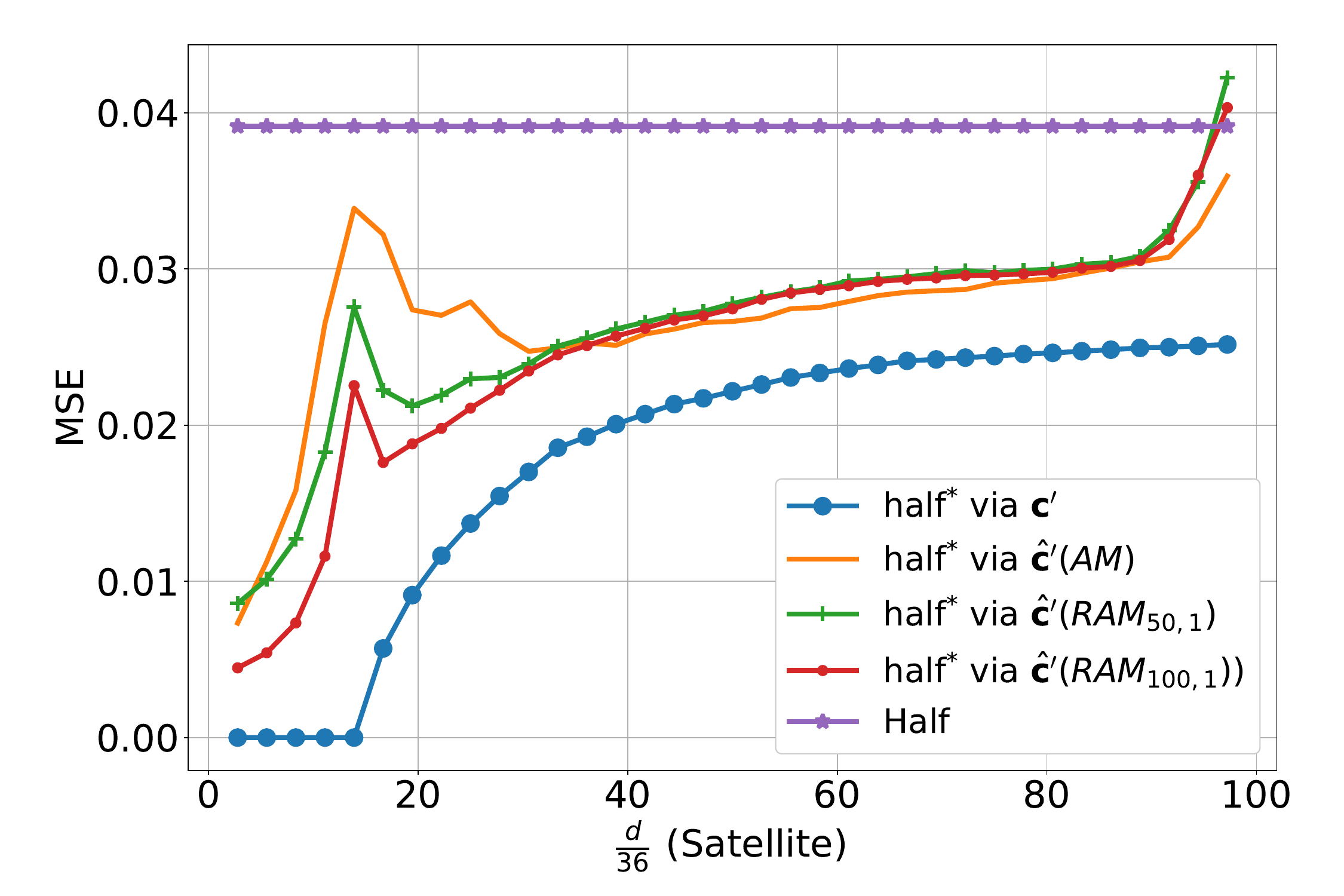}} 
 \caption{MSE per feature resulted from agnostic inference attack via  $\hat{\mathbf{c}}'$ obtained from AM, RAM$_{(n_p, \alpha) = (50,1)}$, RAM$_{(n_p, \alpha) = (100,1)}$, half$^*$ applied $\mathbf{c}'$ and half estimation for the Satellite dataset.}
 \label{fig3} 
\end{figure}

\begin{figure*}[t]
 \centering 
 \scalebox{0.21} 
 {\includegraphics{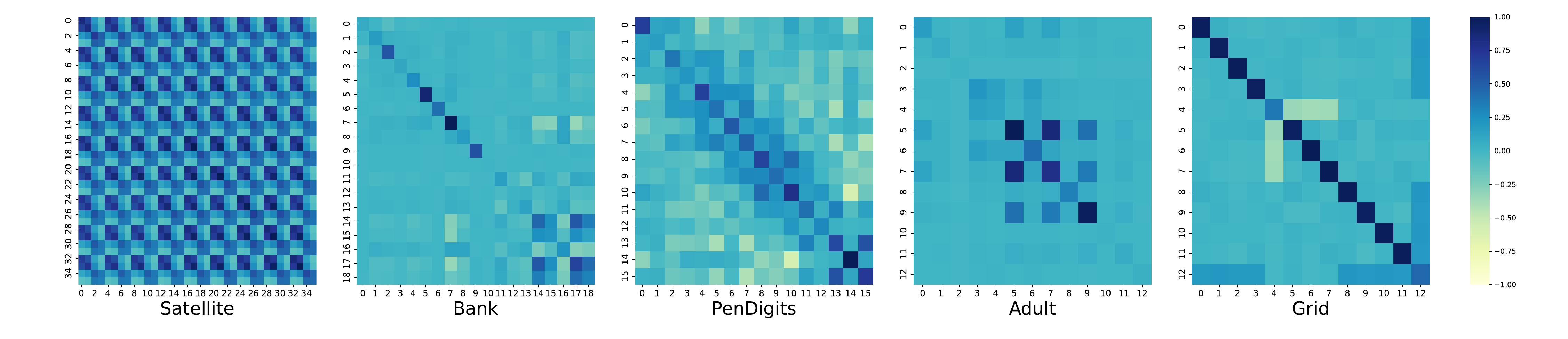}} 
 \caption{Heat map of covariance matrix for the datasets in Table \ref{table_dataset}}
 \label{fig4} 
\end{figure*}

\begin{figure}[ht]
 \centering 
 \scalebox{0.2} {\includegraphics{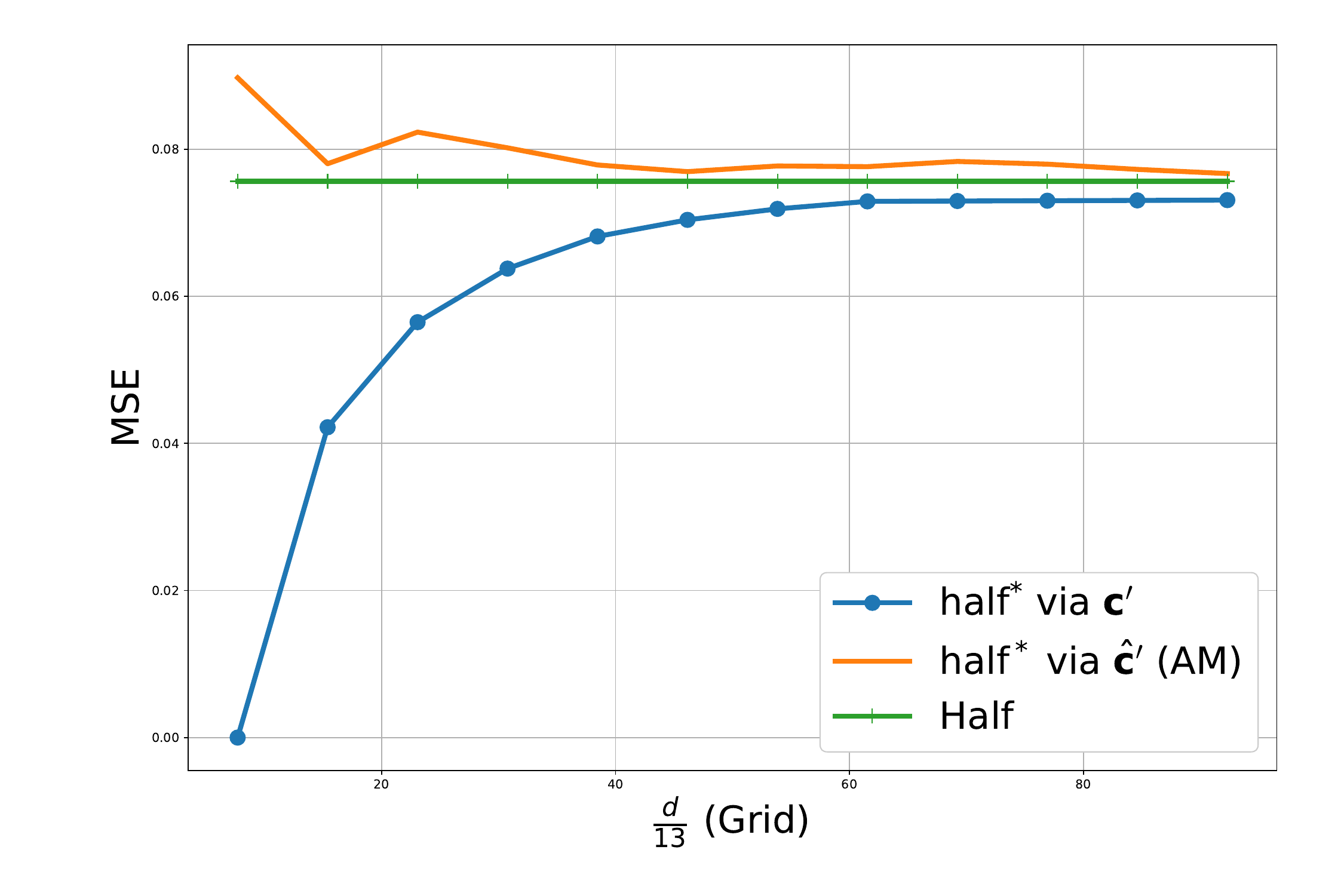}} 
 \caption{MSE per feature resulted from agnostic inference attack via $\hat{\mathbf{c}}'$ obtained from AM, half$^*$ applied on $\mathbf{c}'$ and half estimation.}
 \label{fig5} 
\end{figure}

\begin{figure}[t]
 \centering 
 \scalebox{0.2} 
 {\includegraphics{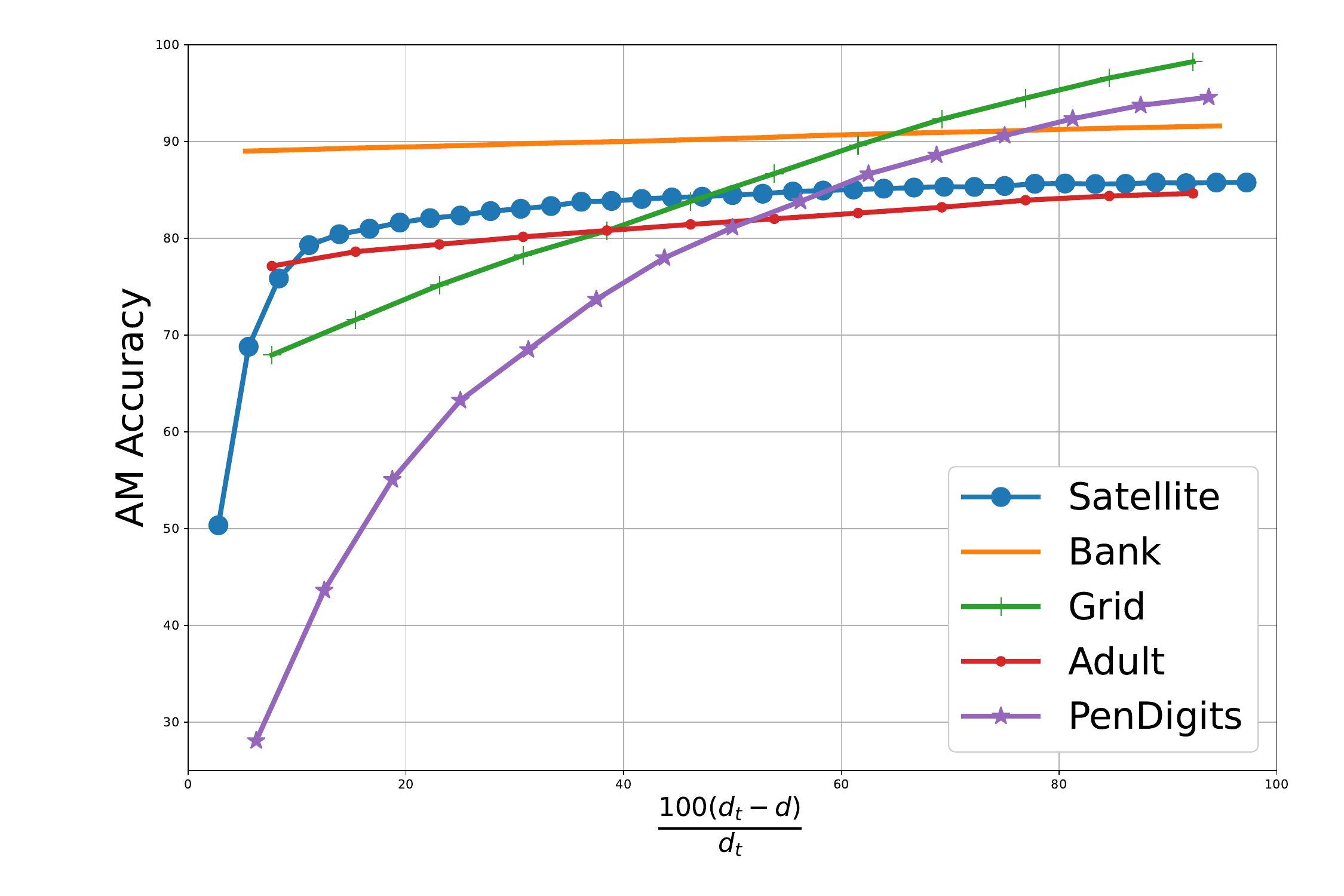}} 
 \caption{Accuracy of AM versus active features for the datasets in Table \ref{table_dataset}}
 \label{fig6} 
\end{figure}

\begin{figure*}[t]
 \centering 
 \scalebox{0.2} 
 {\includegraphics{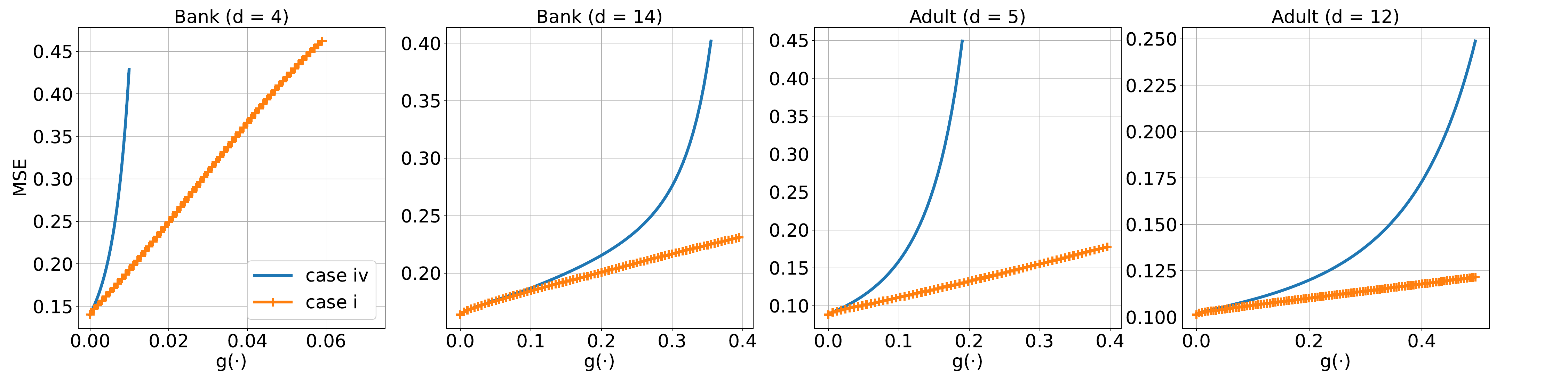}} 
 \caption{Illustration of PI trade-off for case i and case iv for the datasets bank and adult}
 \label{fig8} 
\end{figure*}

\subsection{Experimental results on agnostic inference attack}\label{NR1}
In Figure \ref{fig2}, the resulted MSE from agnostic inference attack via AM (orange solid line denoted by AM) and refined AM (green solid$\pmb{+}$ line denoted by RAM$_{n_p,\alpha}$) are compared with half$^*$ estimation via VFL model (blue solid$\bullet$ line denoted by VFL) and half estimation (red solid$\pmb{\star}$ line) for four different datasets. 
AM is trained based on the loss function defined in (\ref{loss0}) using the training dataset available at the active party, while RAM is trained based on the loss function defined in (\ref{loss1}) with the addition of $n_p =100$ confidence scores in their exact form received from CA.
The RAM in this figure is trained with $\alpha = 1$ as defined in (\ref{loss1}).
As this is an evaluation for agnostic attacks, the MSE is calculated by considering both the training and prediction samples.
It should be noted that since the scores obtained from AM are estimated values of the scores obtained from the VFL model, it is possible that even in the case of $d<k$, the estimated passive features are outside the feasible region (for $d>k$, the system is underdetermined and there is no guarantee that the solution will be in the feasible region).
In such cases, the individual features that are outside the feasible region are mapped to either 0 or 1, whichever is closest.
For the half estimation, all the passive features are estimated as half.
The results show that even without the exact values of the target's confidence scores, it is still possible to achieve a good performance on the inference attack using AM.
For the bank, satellite, and adult datasets, the attack performance using AM is close to the half$^*$ method via the VFL model in which the adversary has access to the exact values of the target's confidence scores.
Moreover, refining AM using (\ref{loss1}) with only 100 scores significantly improves the attack performance, especially when the number of passive features is small. This is mainly because the active party has the advantage of a richer dataset.
Finally, while the attack using AM is not effective for the pendigit dataset, the attack using RAM significantly improves the performance.

\begin{figure*}[!t]
 \centering 
 \scalebox{0.2}{
   \includegraphics{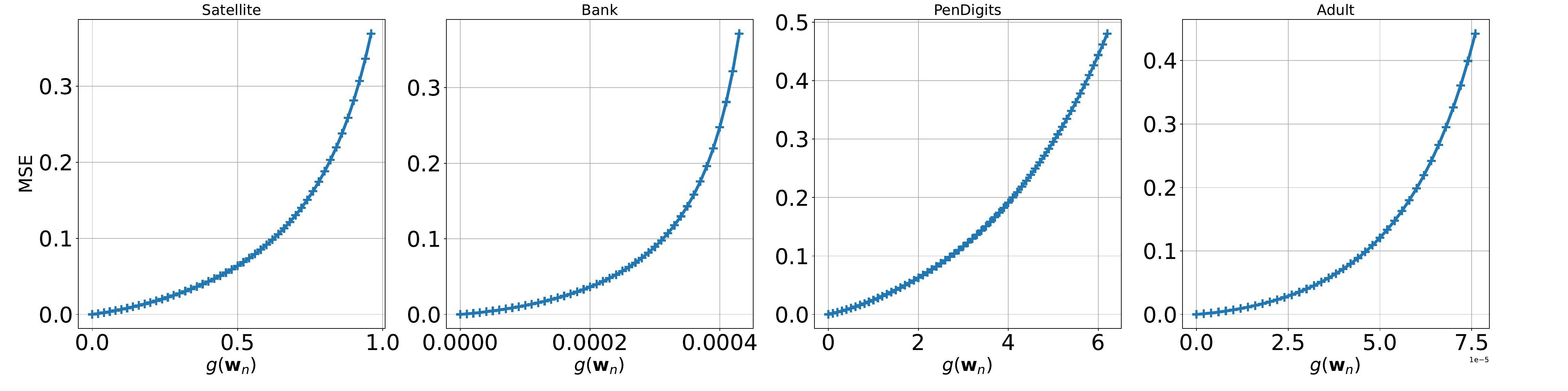} 
 }
 \caption{Illustration of PI trade-off for case iii for the datasets bank, adult, satellite and pendigits.}
 \label{fig7} 
\end{figure*}

\begin{figure*}[!t]
 \centering 
 \scalebox{0.2}{
   \includegraphics{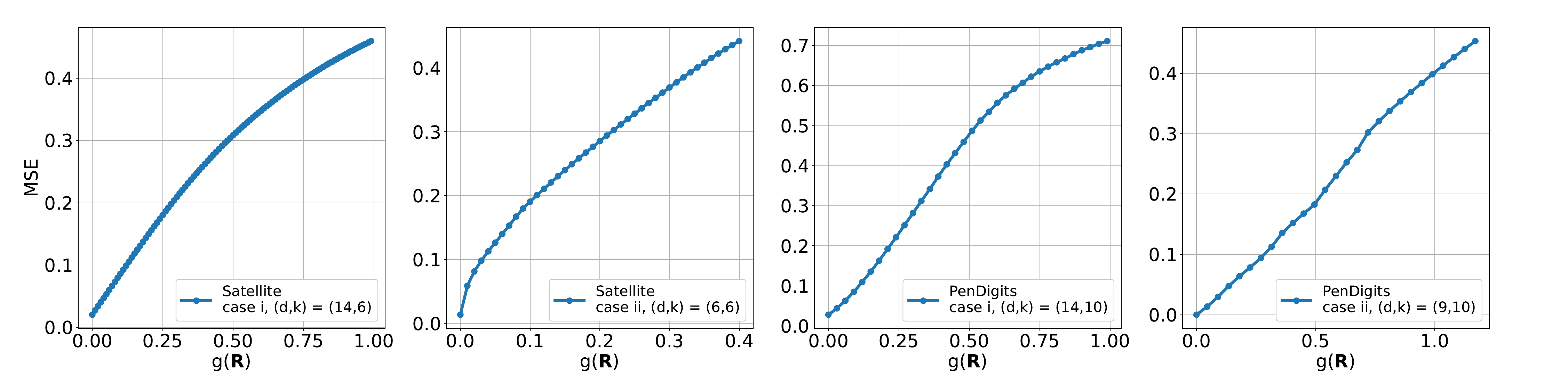} 
 }
 \caption{Illustration of PI trade-off for case i for the datasets satellite and pendigits.}
 \label{fig9} 
\end{figure*}

In Figure \ref{fig3}, performance of agnostic inference attack for satellite dataset when the adversary obtains the estimated confidence scores via AM (orange solid), RAM using 50 prediction confidence scores (green solid$+$) and RAM using 100 prediction confidence scores (red solid$\cdot$) are compared with half estimation and half$^*$ via the VFL model. 
It can be seen that incorporating more confidence scores during the refinement process of AM prior to the attack improves the attack performance. The results show that even using as few as 50 confidence scores can still lead to a significant improvement in the attack performance when the number of passive features is low.


\begin{remark}
The experimental results in this section demonstrate that while the use of confidence scores by the active party can certainly enhance the performance of the inference attack, it is the presence of passive parameters that offers the initial opportunity for exploiting the passive features.
It is important to note that tampering with either the confidence scores or the passive parameters can have adverse effects on either the accuracy of the decisions or the interpretability of the results, respectively. In many scenarios, it may be considered more crucial to maintain accuracy in decision-making than the interpretability of the outcomes. Furthermore, as it will be clearer in section \ref{NR2}, for some cases, a small perturbation of the passive parameters can result in a significant improvement in preserving privacy while still allowing for some level of interpretability in the VFL model and keeping decisions accurate (delivering the confidence scores in the clear).
\end{remark}

\begin{figure*}[t]
 \centering 
 \scalebox{0.2} 
 {\includegraphics{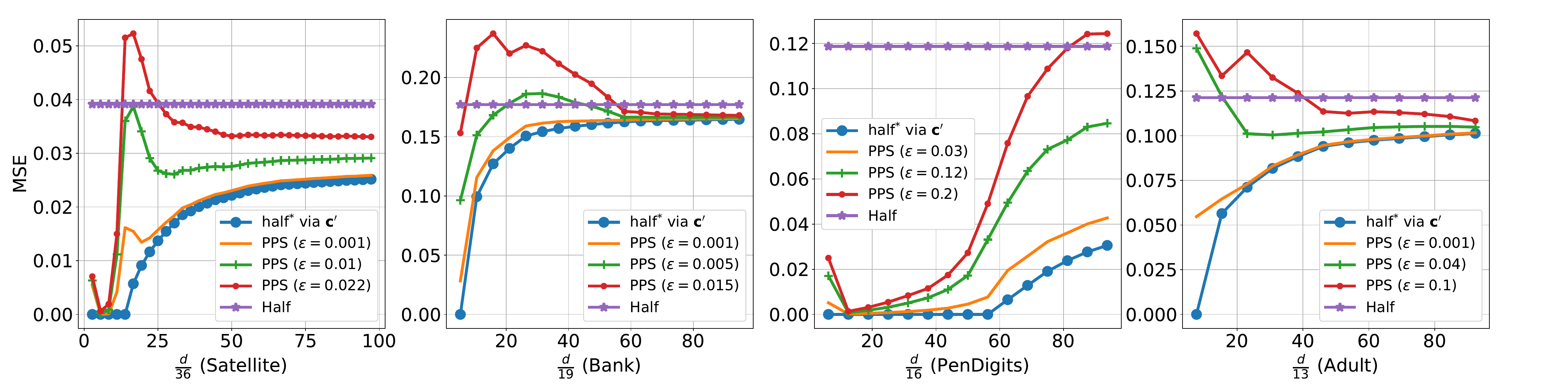}} 
 \caption{MSE resulted from inference attack when PPSs (with different interpretability levels) proposed in section \ref{Privacy_algo} are in place, compared with half$^*$ and half estimation}
 \label{fig11} 
\end{figure*}

\begin{figure}[t]
 \centering 
 \scalebox{0.2} 
 {\includegraphics{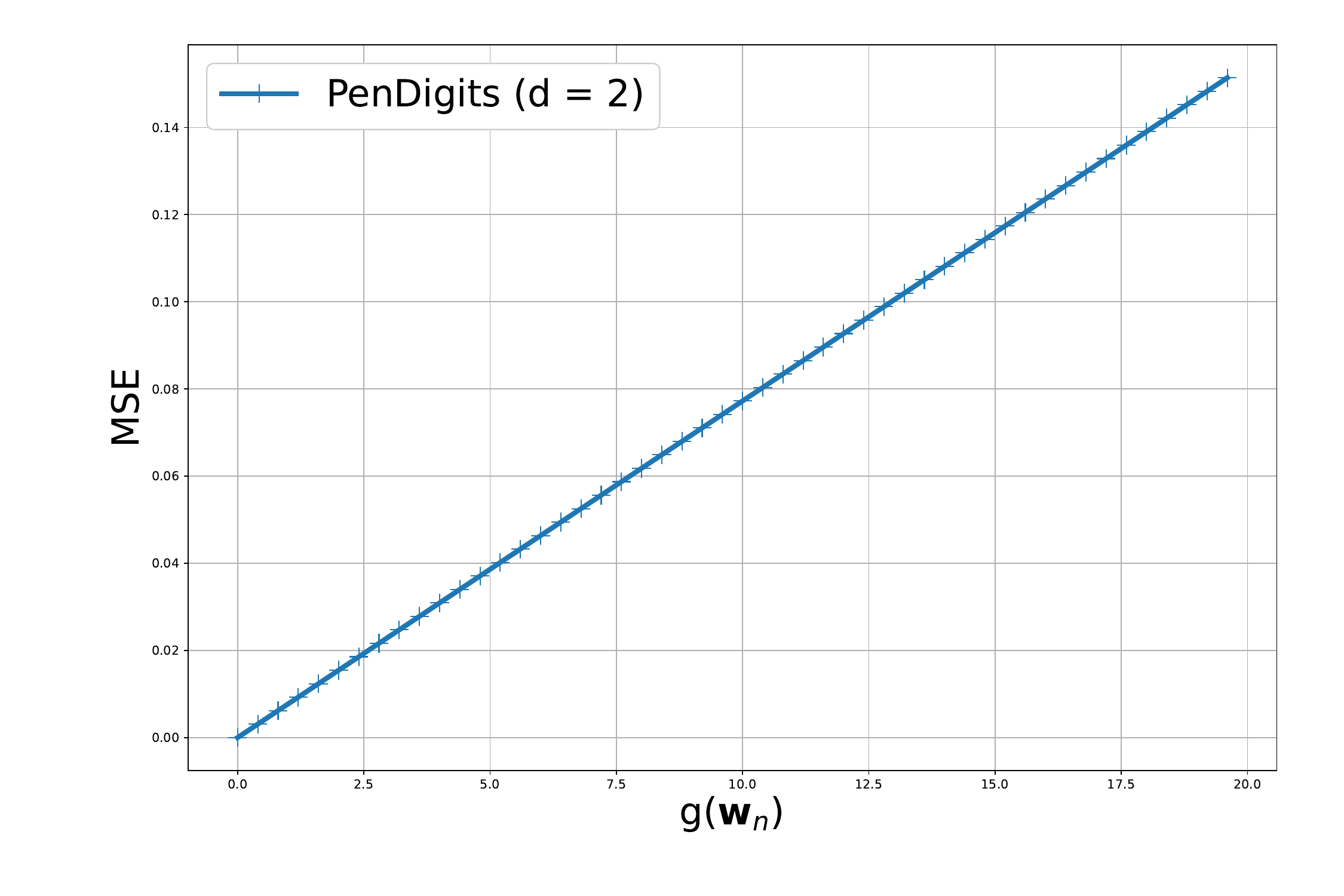}} 
 \caption{MSE per feature of inference attack versus different values of $\varepsilon$ for $d=2$ (Pendigits).}
 \label{fig10} 
\end{figure}

The results shown in Figure \ref{fig2} reveal that the effectiveness of an agnostic inference attack is highly dependent on the nature of the data. For instance, it can be seen that an attack on the bank and adult datasets results in a higher gain compared to the pendigit dataset. This leads to the question of under what conditions the adversary can achieve a lower MSE through this type of inference attack.

We address this question experimentally. We can consider an ideal scenario in which the features of a dataset are mutually independent. In this scenario, training a classifier on one set of features would not reveal any information about the other set of features, thus rendering the inference attack infeasible. However, in real-world situations, there is always a degree of dependency among the features of a dataset.

Due to lack of technical tools to measure the dependency of two or more random variables directly, instead, here we use covariance of the features of a dataset as a an indirect metric. This is represented by the heat map of the covariance matrix $\mathbf{K}_{\pmb{\mu}}$ shown in Figure \ref{fig4}. 
This matrix is calculated as $\mathbf{K}_{\pmb{\mu}}\triangleq \mathds{E}[(\mathbf{X}-\pmb{\mu})(\mathbf{X}-\pmb{\mu})^T]$, where $\mathbf{X}$ is a $d_t$-dimensional vector that represents the features (both passive and active features) and $\pmb{\mu}$ is the mean vector.

It is observed that for the grid dataset, the corresponding covariance matrix is nearly diagonal, indicating that the features are nearly uncorrelated. As shown in Figure \ref{fig5}, the agnostic inference attack on this dataset is ineffective as it is almost similar to estimating the features as half.

\begin{remark}
Figure \ref{fig6} presents the average model accuracy of AM versus the number of active features in the datasets listed in Table \ref{table_dataset}. The results show that the active party benefits more from VFL in the grid and pendigits datasets compared to the bank, satellite and adult datasets. This is because as the number of active features increases, the active party can achieve higher accuracy on the grid and pendigit datasets.
The results from Figure \ref{fig6} in conjunction with the results from Figures \ref{fig2}, \ref{fig4} and \ref{fig5} suggest that datasets that are more suitable for VFL collaboration and result in higher accuracy from VFL collaboration are less susceptible to agnostic inference attacks. This is a crucial aspect to consider for the active party when deciding to participate in a VFL collaboration, especially if the active party incurs costs for both VFL training and prediction. 
\end{remark}

\subsection{Experimental results on privacy-preserving schemes}\label{NR2}
In this section, we evaluate the effectiveness of the PPSs introduced in section \ref{PPS}. Our assessment consists of two stages: first, we analyze each of the cases studied in section \ref{Privacy_algo} separately by keeping either $d$ and/or $k$ constant and varying the value of $\varepsilon$. Next, we examine the opposite scenario by fixing $\varepsilon$ and varying $d$ and $k$.

To measure the interpretability of the passive parameters at the active party, we use the function $g(\cdot)$, where a higher value of $g(\cdot)$ indicates lower interpretability. On the other hand, the privacy of the passive party is measured by the amount of MSE incurred by the active party in the case of an attack, with higher MSE values indicating greater privacy for the passive party.

In Figure \ref{fig8}, the PI trade-offs for the datasets bank and adult are illustrated for three different values of passive features. 
The orange-solid$+$ and blue-solid$\bullet$ lines are the PI trade-offs obtained via the approaches in case i and case iv, respectively. 
It is worth noting two key aspects regarding the results. 
\textit{First}, the inclusion of the results of case i is merely intended to provide a comparison with the more general PPS approach in case iv, and they are not applicable for datasets with $k=2$. 
As seen from the plots, while the results of case i deviate quickly from those of case iv, they remain close to each other in the region of interest, which is approximately 0.18 and 0.14 for the bank and adult datasets, respectively, as indicated in Figure \ref{fig2} (see the MSE obtained from half estimation). Furthermore, as the value of $d$ increases, the results from case i and case iv converge, as shown in the rightmost plots in Figure \ref{fig8}.
\textit{Second}, both the passive and active parties prefer that a small amount of perturbation in the passive parameters results in a noticeable increase in the MSE from inference attacks, while still preserving a reasonable level of interpretability. 
This way, both parties meet their concerns with a minimal passive parameter perturbation. This is equivalent to having a concave PI trade-off curve. However, as observed from Figure \ref{fig8} and the results throughout this section, the PI trade-off achieved by the proposed PPSs for different datasets and under different cases can sometimes be convex, sometimes concave, and in some cases, neither. It remains an open question whether the optimal PI trade-off curve in the region of interest follows a certain shape, preferably concave.

In Figure \ref{fig7}, the PI trade-off is demonstrated for the datasets bank, adult, satellite, and pendigits, when $d=1$. To obtain the PI curve for the bank and adult datasets, we utilized the result from Remark (\ref{rem100}), while for the remaining datasets, the optimization problem in (\ref{M32}) was solved. The approach in this case is considered more comprehensive as it does not impose any limitations on the mapping between the passive parameters and the parameters that the active party has given access to. Note that, as shown in Figure \ref{fig2}, the MSE for the half estimation on the satellite, bank, pendigits, and adult datasets is approximately 0.04, 0.17, 0.12, and 0.12, respectively. This leads to the observation that the PI curve below this MSE level for each of these datasets is almost linear. This observation is promising in the sense that the PPS is capable of delivering a balance between privacy and interpretability concerns of the parties, despite the general non-concave PI trade-off curve.

In Figure \ref{fig9}, the PI trade-off for case i and ii is depicted for the satellite and pendigits datasets. It's important to keep in mind that in this instance, there is no widely accepted general solution that can serve as a benchmark for comparison. As a result, it remains a potential area of research to identify mappings from $\mathbf{W}_\textnormal{pas}$ to $\mathbf{W}_n$ that result in an improvement in the PI trade-off region. 

In Figure \ref{fig11}, the effectiveness of the PPSs outlined in section \ref{Privacy_algo} is compared to the baseline methods half$^*$ and half estimation. Each plot in the figure corresponds to one of the datasets, with three different levels of interpretability depicted in each plot. For instance, in the top-left plot related to the satellite dataset, the interpretability levels are $\varepsilon = 0.001, 0.01, 0.022$. Unlike Figures \ref{fig8}, \ref{fig7}, and \ref{fig9}, in this figure the performance of all four cases for PPS are evaluated against varying numbers of passive features with a fixed interpretability level in each curve. The results indicate that, for a fixed $\varepsilon$, the gap between the MSE and that of half$^*$ is not consistent. This is because as the number of passive parameters increases, finding the optimal solution for the system $\mathbf{A}_n\mathbf{X}=\mathbf{b}'$ becomes more challenging. For the satellite and pendigit datasets, the MSE does not seem to be affected as the value of $\varepsilon$ varies for the case of $d=2$. However, as seen in the results in Figure \ref{fig10} (MSE per feature of inference attack versus different values of $\varepsilon$ for $d=2$), the amount of $\varepsilon$ required to produce a noticeable change in the MSE is not comparable to the other cases.

\begin{remark}\label{rem:DP_Pri}
(On baseline defenses and experimental scope):
Throughout this paper, we adopt a worst-case utility assumption in which the active party receives confidence scores in their exact form.
This choice is motivated by applications that rely on soft scores for downstream tasks and allows us to study privacy risks that cannot be mitigated by score-level obfuscation.
Consequently, defenses based on rounding, clipping, or injecting noise into confidence scores are intentionally excluded from the experimental evaluation. While adding random noise to model parameters may appear as a simple alternative, such perturbations lack a principled design criterion unless calibrated through an explicit objective.
In contrast, the proposed PPSs introduce structured parameter distortions obtained via optimization, enabling a controlled and interpretable privacy--interpretability trade-off.
Black-box configurations, where passive parameters are entirely hidden, trivially eliminate reconstruction attacks but also remove interpretability at the active party; they therefore serve as conceptual reference points rather than experimental baselines.
\end{remark}

\begin{remark}\label{rem:Gaus_baseline}
(On Gaussian noise baselines): A natural baseline is to perturb the passive-party parameters with i.i.d.\ Gaussian noise under the same distortion budget $\varepsilon$. However, in the present setting, such a comparison is not particularly informative. First, the passive parameter block is low-dimensional, and therefore the induced privacy metric (reconstruction MSE) is highly sensitive to individual noise realizations. Unlike our optimization-based PPS, which is deterministic once $\varepsilon$ is fixed, a naive Gaussian mechanism yields unstable privacy--interpretability outcomes unless one averages over many independent perturbation draws. Second, our evaluation already involves averaging over samples and over different passive-feature allocations. A fair Gaussian comparison would therefore require an additional Monte Carlo averaging layer, substantially increasing variance reduction requirements and computational cost. Moreover, the Gaussian variance cannot be directly matched to the same distortion budget $\varepsilon$ in closed form; instead, it must be calibrated empirically through iterative procedures (e.g., bisection), where each candidate variance requires repeated evaluations of the full pipeline. 
Finally, Gaussian perturbation is geometry-agnostic and does not exploit the structure of the reconstruction system induced by the passive parameters. As a result, in this low-dimensional regime, its outcomes are weakly diagnostic: observed differences are often dominated by stochastic variability rather than reflecting meaningful privacy--utility trade-offs. For these reasons, we do not include this baseline and instead focus on the proposed PPS, which directly optimizes distortion under the reconstruction model.
\end{remark}

\subsection{Overhead Analysis and Practical Considerations}
\label{subsec:overhead}

In the following, we discuss the system-level overhead introduced by the proposed PPSs, with particular emphasis on their practical implementation and deployment cost.

\textbf{Computational overhead}:
All PPSs proposed in this paper are applied \emph{offline}, after completion of VFL model training and prior to releasing the passive party’s parameters to the active party.
For PPS designs formulated as optimization problems—including those with orthogonality constraints on Stiefel manifolds—the optimization is carried out locally at the passive party and involves only the learned model parameters and precomputed second-order statistics.

In our implementation, the PPS optimization is implemented in Python using standard numerical libraries (NumPy and SciPy), together with manifold-optimization tools for orthogonality-constrained problems.
In particular, Stiefel-manifold-based designs are solved using existing manifold solvers (e.g., smooth projected or Riemannian gradient-based methods), which operate directly on matrix variables and exploit closed-form expressions for the objective and its gradient.
Because the optimization variables are of dimension proportional to the number of passive features and classes, and because no data-dependent iterations are involved, the resulting wall-clock runtime is on the order of milliseconds for all datasets considered.

\textbf{Communication overhead}:
The proposed PPSs do not introduce additional communication rounds during VFL training or prediction.
The only additional step is the one-time release of the modified passive parameters in place of the original ones.
Therefore, from a system perspective, the communication pattern of the underlying VFL protocol remains unchanged.

\subsection{Guidelines for VFL system designers}\label{sec:guidelines}
At a high level, agnostic attacks are most concerning in deployments where (i) the active party has access to labels, (ii) features across parties exhibit moderate or strong correlation, and (iii) the application requires soft confidence scores for downstream tasks.
In such settings, the active party can derive surrogate models (AM/RAM) that enable reconstructions comparable in accuracy to those obtained under non-agnostic attacks.

The proposed PPSs are intended for scenarios where passive-party interpretability is permitted or required, but where exact parameter disclosure poses privacy risks.
They are applied once after model training and adjust only the passive-party parameters, allowing the VFL workflow and prediction interface to remain unchanged.
For practitioners, the key operational guideline is to select the distortion level that achieves the desired privacy target (in terms of reconstruction MSE) while retaining sufficient interpretability for audit or compliance purposes.

\section{Conclusion}\label{conc}
In this paper, we examine the intricacies of a white-box VFL setting where two parties are involved in constructing a classifier collaboratively. The active party has access to the labels while the passive party contributes by sharing a separate set of features for the same samples. This collaboration takes place under the supervision of a third trusted entity known as the CA.
Previous studies, such as those in \cite{RaVaGu, Xinjian}, have explored various inference attack methods that mainly rely on the active party's exploitation of the confidence scores during the prediction phase.
However, in this study, we investigate a novel approach of inference attack, which does not require the active party to have access to the target's confidence score. The adversary uses their available training set, including labels and active party features, to construct a classifier that is then used to estimate the target's confidence score. Once the estimation is obtained, the active party uses one of the previously studied inference attack methods to carry out the attack. This attack is referred to as ``agnostic'' due to the absence of the target's confidence score at the active party, which as a result puts all the samples used in the training phase at the risk of privacy exposure.
Our findings show that the performance of the adversary's model can be significantly enhanced by incorporating the confidence scores received from CA. This results in more accurate estimates of the target's confidence score, which, in turn, can lead to more successful inference attacks.
Since this type of attack does not rely on the availability of an attack target's confidence score, it is therefore unaffected of any PPSs applied to the confidence scores, such as rounding or adding noise. 

To counter the aforementioned potential for privacy breaches in a VFL setting, various privacy-preserving schemes (PPSs) have been proposed with a focus on protecting the passive parameters.
To that end, we consider the development of PPSs that intentionally distort the passive parameters in a systematic manner.
The active party is given access to these distorted parameters instead of the original versions.
In designing these PPSs, we carefully consider the need for interpretability by the active party while still safeguarding the privacy of the passive party.
This, in turn, gives rise to a trade-off between the interpretability of the VFL model for the active party and the privacy concerns of the passive party.
Our approach moves away from a black or white viewpoint in VFL and instead aims to find a suitable balance between the needs of both parties to mitigate the risk of collaboration.
Our goal is not to completely obscure the active party or completely protect the passive party, but rather to strike a balance between the two in order to minimize potential privacy risks.

\appendix
\section{Derivation of (\ref{M18})}\label{app:1}
\begin{align}
    \textnormal{MSE}(\hat{\mathbf{X}}_{\textnormal{LS},a})&=\frac{1}{d}\mathds{E}\left[\|\mathbf{X}-\hat{\mathbf{X}}_{\textnormal{LS},a}\|^2\right]\nonumber\\
    &=\frac{1}{d} \mathds{E}\left[\|\mathbf{X}-(\mathbf{A}^T\mathbf{A})^{-1}\mathbf{A}^T(\mathbf{b}'+\hat{\mathbf{c}}'-\mathbf{c}')\|^2\right]\nonumber\\
    &= \frac{1}{d}\mathds{E}\left[\|\mathbf{X}-(\mathbf{A}^T\mathbf{A})^{-1}\mathbf{A}^T(\mathbf{AX}+\hat{\mathbf{c}}'-\mathbf{c}')\|^2\right]\nonumber\\
    &=\frac{1}{d}\mathds{E}\left[\|(\mathbf{A}^T\mathbf{A})^{-1}\mathbf{A}^T(\hat{\mathbf{c}}'-\mathbf{c}')\|^2\right]\nonumber\\
    &=\frac{1}{d} \mathds{E}\left[\textnormal{Tr}(\mathbf{A}(\mathbf{A}^T\mathbf{A})^{-1}(\mathbf{A}^T\mathbf{A})^{-1}\mathbf{A}^T(\hat{\mathbf{c}}'-\mathbf{c}')(\hat{\mathbf{c}}'-\mathbf{c}')^T\right]\nonumber\\
    &=\frac{1}{d}\textnormal{Tr}\left(\mathbf{A}(\mathbf{A}^T\mathbf{A})^{-2}\mathbf{A}^T\mathbf{K}_{\mathbf{c}',\mathbf{\hat{c}}'}\right),
\end{align}

\section{Derivation of (\ref{M2})}\label{app:2}
\begin{align}
\textnormal{MSE}(\hat{\mathbf{X}}_{\textnormal{half}^*, a})=& \frac{1}{d}\mathds{E}\left[\|\mathbf{X}-\hat{\mathbf{X}}_{\textnormal{half}^*, a}\|^2\right]\nonumber\\
=& \frac{1}{d} \mathds{E}\left[\|(\mathbf{X}-\hat{\mathbf{X}}_{\textnormal{half}^*}-\mathbf{A^+(\mathbf{\hat{c}}'-\mathbf{c}')})\|^2\right]\nonumber\\
=&\frac{1}{d}\mathds{E}\left[\|\mathbf{X}-\hat{\mathbf{X}}_{\textnormal{half}^*}\|^2\right] + \frac{1}{d}\mathds{E}\left[\left\|\mathbf{A}^+(\mathbf{\hat{c}}'-\mathbf{c}')\right\|^2\right]-\frac{2}{d}\mathds{E}\left[(\mathbf{\hat{c}}'-\mathbf{c}')^T{\mathbf{A}^{+}}^{T}(\mathbf{I}-\mathbf{A}^+\mathbf{A})(\mathbf{X}-\frac{1}{2}\mathbf{1}_d)\right]\label{M21}\\
=&\frac{1}{d}\textnormal{Tr}\left(\left(\mathbf{I}-\mathbf{A}^+\mathbf{A}\right)\mathbf{K}_{\frac{1}{2}\mathbf{1}}\right)+\frac{1}{d}\mathds{E}\left[\left\|\mathbf{A}^+(\mathbf{\hat{c}}'-\mathbf{c}')\right\|^2\right]\label{M9} \\
=&\frac{1}{d}\textnormal{Tr}\left(\left(\mathbf{I}-\mathbf{A}^+\mathbf{A}\right)\mathbf{K}_{\frac{1}{2}\mathbf{1}}\right)+ \frac{1}{d}\textnormal{Tr}\left({\mathbf{A}^+}^T\mathbf{A}^{+}\mathbf{K}_{\mathbf{c}',\mathbf{\hat{c}}'}\right),
\end{align}
where the third term in (\ref{M21}) is $0$ due to ${\mathbf{A}^{+}}^{T}(\mathbf{I}-\mathbf{A}^+\mathbf{A})=\mathbf{0}$. The first term in (\ref{M9}) is due to \cite[Theorem 2]{RaVaGu}. It is important to note that the second term in (\ref{M21}) is always positive, resulting in a positive value for the second term in (\ref{M2}).  

\section{Derivation of (\ref{M26})}\label{app:3}
\begin{align}
   \textnormal{MSE}(\hat{\mathbf{X}}_\textnormal{LS}^\textnormal{PPS})&=\mathds{E}[\|\mathbf{X}-\hat{\mathbf{X}}_\textnormal{LS}^\textnormal{PPS}\|^2] \nonumber\\
   &=\mathds{E}[\|\mathbf{X}-\mathbf{A}_n^{+}\mathbf{b}'\|^2]\nonumber\\
   &=\mathds{E}[\|(\mathbf{I}-\mathbf{A}_n^{+}\mathbf{A})\mathbf{X}\|^2]\label{M22}\\
   &=\mathds{E}[\|(\mathbf{I}-\mathbf{R}^T\mathbf{A}^{+}\mathbf{A})\mathbf{X}\|^2]\label{M23}\\
   &=\textnormal{Tr}((\mathbf{I}-\mathbf{R}^T\mathbf{A}^{+}\mathbf{A})\mathbf{K}_\mathbf{0}(\mathbf{I}-\mathbf{A}^{+}\mathbf{A}\mathbf{R}))\label{M24}\\
   &=\textnormal{Tr}((\mathbf{I}+\mathbf{A}^{+}\mathbf{A}-\mathbf{A}^{+}\mathbf{A}\mathbf{R}-\mathbf{R}^T\mathbf{A}^{+}\mathbf{A})\mathbf{K}_\mathbf{0})\label{M25}\\
   &=\textnormal{Tr}(\mathbf{I}+\mathbf{A}^{+}\mathbf{A})-2\textnormal{Tr}(\mathbf{R}\mathbf{K}_\mathbf{0}\mathbf{A}^{+}\mathbf{A}),\label{M26_1}
\end{align}
where in (\ref{M22}), (\ref{M23}), (\ref{M24}), (\ref{M25}), (\ref{M26_1}) we have 
$\mathbf{b}' = \mathbf{AX}$, 
$(\mathbf{AR})^{+} = \mathbf{R}^{-1}\mathbf{A}^{+}=\mathbf{R}^{T}\mathbf{A}^{+}$, because of orthonormality of $\mathbf{R}$, 
$(\mathbf{A}^{+}\mathbf{A})^2 = \mathbf{A}^{+}\mathbf{A}$, 
$\mathbf{K}_\mathbf{\pmb{\mu}} \triangleq \mathds{E}[(\mathbf{X}-\pmb{\mu})(\mathbf{X}-\pmb{\mu})^T]$ with $\pmb{\mu} = \mathbf{0}$, 
$\textnormal{Tr}(\mathbf{R}^T\mathbf{A}^{+}\mathbf{AK}_\mathbf{0}) = \textnormal{Tr}(\mathbf{A}^{+}\mathbf{ARK}_\mathbf{0})$, respectively.

\section{Derivation of (\ref{eq1})}\label{app:4}
\begin{align}\nonumber
    \textnormal{MSE}(\hat{\mathbf{X}}_{\textnormal{half}^*,a}^{\textnormal{PPS}})  &= \frac{1}{d}\mathds{E}\left[\|\mathbf{X} -\hat{\mathbf{X}}_{\textnormal{half}^*,a}^{\textnormal{PPS}}\|^2\right]\\\nonumber
    &= \frac{1}{d}\mathds{E}\left[\|\mathbf{X} - \mathbf{R}^T\mathbf{A}^{+}\mathbf{AX} - \frac{1}{2}(\mathbf{I}-\mathbf{R}^T\mathbf{A}^{+}\mathbf{AR})\mathbf{1}_d\|^2\right]\\\nonumber
    &=\frac{1}{d}\mathds{E}\left[\|\mathbf{X}-\frac{1}{2}\mathbf{1}_d-\mathbf{R}^T\mathbf{A}^{+}\mathbf{A}(\mathbf{X}-\frac{1}{2}\mathbf{R}\mathbf{1}_d)\|^2\right]\\\nonumber
    &=\frac{1}{d}\mathds{E}\left[\|\mathbf{X}-\frac{1}{2}\mathbf{1}_d\|^2+\textnormal{Tr}(\mathbf{A}^{+}\mathbf{A}(\mathbf{X}-\frac{1}{2}\mathbf{R}\mathbf{1}_d)(\mathbf{X}-\frac{1}{2}\mathbf{R}\mathbf{1}_d)^T)-2\textnormal{Tr}(\mathbf{R}^T\mathbf{A}^{+}\mathbf{A}(\mathbf{X}-\frac{1}{2}\mathbf{R}\mathbf{1}_d)(\mathbf{X}-\frac{1}{2}\mathbf{1}_d)^T) \right]\label{eq0}\\
    &= \frac{1}{d}\textnormal{Tr}(\mathbf{K}_{\frac{1}{2}\mathbf{1}})+\frac{1}{d}\mathds{E}\left[\textnormal{Tr}(\mathbf{A}^{+}\mathbf{A}(\mathbf{X}-\frac{1}{2}\mathbf{R}\mathbf{1}_d)\cdot(\mathbf{X}+\frac{1}{2}\mathbf{R}\mathbf{1}_d-2\mathbf{RX})^T)\right]\\
    &= \frac{1}{d}\textnormal{Tr}(\mathbf{K}_{\frac{1}{2}\mathbf{1}})+\frac{1}{d}\mathds{E}\left[\textnormal{Tr}(\mathbf{A}^{+}\mathbf{A}(\mathbf{XX}^T-2\mathbf{XX}^T\mathbf{R}^T))\right]-\frac{1}{d}\mathds{E}\left[\textnormal{Tr}(\mathbf{A}^{+}\mathbf{A}(\frac{1}{4}\mathbf{R}\mathbf{1}_d\mathbf{1}_d^T\mathbf{R}^T-\mathbf{R1}_d\mathbf{X}^T\mathbf{R}) ) \right]\label{eq2}\\\nonumber
    &= \frac{1}{d}\textnormal{Tr}\left(\mathbf{K}_{\frac{1}{2}\mathbf{1}}+\mathbf{A}^{+}\mathbf{A}\left(\mathbf{K}_\mathbf{0}-2\mathbf{K}_\mathbf{0}\mathbf{R}^T\frac{1}{4}\mathbf{R}\mathbf{1}_d\mathbf{1}_d^T\mathbf{R}^T+\mathbf{R1}_d\pmb{\mu}^T\mathbf{R}\right)\right )\nonumber\\
    &= \frac{1}{d}\textnormal{Tr}(\mathbf{K}_{\frac{1}{2}\mathbf{1}})-\frac{1}{d}\textnormal{Tr}\left(\mathbf{A}^{+}\mathbf{A}\left(2\mathbf{K}_\mathbf{0}\mathbf{R}^T-\mathbf{K}_\mathbf{0}-\mathbf{RMR}^T\right) \right)\label{eq1_1},
\end{align}
where in (\ref{eq0}) we use $\mathbf{R}^T\mathbf{R} = \mathbf{I}$, and $(\mathbf{A}^{+}\mathbf{A})^T\mathbf{A}^{+}\mathbf{A} =(\mathbf{A}^{+}\mathbf{A})^2=\mathbf{A}^{+}\mathbf{A}$. In (\ref{eq2}). we use $\textnormal{Tr}(\mathbf{A}^{+}\mathbf{AX1}_d^T\mathbf{R}^T) = \textnormal{Tr} (\mathbf{R1}_d\mathbf{X}^T\mathbf{A}^{+}\mathbf{A}) = \textnormal{Tr}(\mathbf{A}^{+}\mathbf{AR1}_d\mathbf{X}^T)$
and in (\ref{eq1_1}), we have $\mathbf{M}\triangleq\mathbf{1}_d(\pmb{\mu}-\frac{1}{4}\mathbf{1}_d)^T$.

\section{Derivation of (\ref{M15})}\label{app:5}
\begin{align}
    \textnormal{MSE}(\hat{\mathbf{X}}_{\textnormal{LS},a}^{\textnormal{PPS}})&= \frac{1}{d}\mathds{E}\left[\left\|\mathbf{X}-(\mathbf{A}_n^T\mathbf{A}_n)^{-1}\mathbf{A}_n^T\mathbf{b}'\right\|^2\right]\nonumber \\
    &= \frac{1}{d}\mathds{E}\left[\left\|\mathbf{X}-(\mathbf{A}_n^T\mathbf{A}_n)^{-1}\mathbf{A}_n^T\mathbf{AX}\right\|^2\right]\nonumber\\
    &= \frac{1}{d}\mathds{E}\left[\left\|\mathbf{X}-(\mathbf{R}^T\mathbf{A}^T\mathbf{A}\mathbf{R})^{-1}\mathbf{R}^T\mathbf{A}^T\mathbf{AX}\right\|^2\right]\nonumber\\
    &= \frac{1}{d}\mathds{E}\left[\left\|(\mathbf{I}-\mathbf{R}^T)\mathbf{X}\right\|^2\right]\nonumber\\\label{M15_1}
    &= \frac{2}{d}\textnormal{Tr}\left(\mathbf{K}_\mathbf{0}-\mathbf{R}\mathbf{K}_\mathbf{0}\right).
\end{align}
In (\ref{M15_1}), we use the identity $\textnormal{Tr}(\mathbf{R}^T\mathbf{K}_0) = \textnormal{Tr}(\mathbf{RK}_0)$.

\bibliography{REFERENCE}

@article{penrose, title={A generalized inverse for matrices}, volume={51}, DOI={10.1017/S0305004100030401}, number={3}, journal={Mathematical Proceedings of the Cambridge Philosophical Society}, publisher={Cambridge University Press}, author={Penrose, R.}, year={1955}, pages={406–413}}

@INPROCEEDINGS{Wu2018PPGNN,
  author    = {Wu, Y. and Wang, K. and Zhang, Z. and Lin, W. and Chen, H. and Li, C.},
  booktitle = {Proceedings of the International Conference on Extending Database Technology (EDBT)},
  title     = {Privacy-Preserving Group Nearest Neighbor Search},
  year      = {2018},
  pages     = {277--288}
}

@INPROCEEDINGS{Mohassel2017SecureML,
  author    = {Mohassel, Payman and Zhang, Yupeng},
  booktitle = {Proceedings of the IEEE Symposium on Security and Privacy},
  title     = {SecureML: A System for Scalable Privacy-Preserving Machine Learning},
  year      = {2017},
  pages     = {19--38},
  doi       = {10.1109/SP.2017.12}
}

@INPROCEEDINGS{Chowdhury2020Crypt,
  author    = {Chowdhury, A. R. and Wang, C. and He, X. and Machanavajjhala, A. and S. Jha},
  booktitle = {Proceedings of the International Conference on Management of Data (SIGMOD)},
  title     = {Crypt?: Crypto-Assisted Differential Privacy on Untrusted Servers},
  year      = {2020},
  pages     = {603--619},
  doi       = {10.1145/3318464.3389769}
}

@misc{RaVaGu,
  doi = {10.48550/ARXIV.2207.11788},
  
  url = {https://arxiv.org/abs/2207.11788},
  
  author = {Rassouli, Borzoo and Varasteh, Morteza and Gunduz, Deniz},
  
  title = {Privacy Against Inference Attacks in Vertical Federated Learning},
  
  publisher = {arXiv},
  
  year = {2022},
  
  copyright = {arXiv.org perpetual, non-exclusive license}
}

@article{Jiang,
author = {Jiang, Xue and Zhou, Xuebing and Grossklags, Jens},
year = {2022},
month = {04},
pages = {263-281},
title = {Comprehensive Analysis of Privacy Leakage in Vertical Federated Learning During Prediction},
volume = {2022},
journal = {Proceedings on Privacy Enhancing Technologies},
doi = {10.2478/popets-2022-0045}
}

@misc{Matrixcookbook,
  author = {Petersen, K. B. and Pedersen, M. S.},
  month = oct,
  note = {Version 20081110},
  publisher = {Technical University of Denmark},
  review = {Matrix Cookbook},
  title = {The Matrix Cookbook},
  url = {http://www2.imm.dtu.dk/pubdb/p.php?3274},
  year = 2008
}

@article{Aplus,
author = {Greville, T. N. E.},
title = {Note on the Generalized Inverse of a Matrix Product},
journal = {SIAM Review},
volume = {8},
number = {4},
pages = {518-521},
year = {1966},
doi = {10.1137/1008107},

URL = { 
    
        https://doi.org/10.1137/1008107
    
    

},
eprint = { 
    
        https://doi.org/10.1137/1008107
    
    

}

}

@article{slpg2,
author = {Nachuan Xiao and Xin Liu and Ya-xiang Yuan},
title = {A class of smooth exact penalty function methods for optimization problems with orthogonality constraints},
journal = {Optimization Methods and Software},
volume = {37},
number = {4},
pages = {1205-1241},
year  = {2022},
publisher = {Taylor & Francis},
doi = {10.1080/10556788.2020.1852236},

URL = { 
    
        https://doi.org/10.1080/10556788.2020.1852236
    
    

},
eprint = { 
    
        https://doi.org/10.1080/10556788.2020.1852236
    
    

}

}

@article{slpg3,
author = {Edelman, Alan and Arias, Tom\'{a}s A. and Smith, Steven T.},
title = {The Geometry of Algorithms with Orthogonality Constraints},
journal = {SIAM Journal on Matrix Analysis and Applications},
volume = {20},
number = {2},
pages = {303-353},
year = {1998},
doi = {10.1137/S0895479895290954},

URL = { 
    
        https://doi.org/10.1137/S0895479895290954
},
eprint = { 
    
        https://doi.org/10.1137/S0895479895290954
}
}

@misc{slpg1,
  doi = {10.48550/ARXIV.2103.03514},
  
  url = {https://arxiv.org/abs/2103.03514},
  
  author = {Xiao, Nachuan and Liu, Xin and Yuan, Ya-xiang},
  
  title = {A Penalty-free Infeasible Approach for a Class of Nonsmooth Optimization Problems over the Stiefel Manifold},
  
  publisher = {arXiv},
  
  year = {2021},
  
  copyright = {arXiv.org perpetual, non-exclusive license}
}

@article{Xinjian,
  author    = {Xinjian Luo and
               Yuncheng Wu and
               Xiaokui Xiao and
               Beng Chin Ooi},
  title     = {Feature Inference Attack on Model Predictions in Vertical Federated
               Learning},
  journal   = {CoRR},
  volume    = {abs/2010.10152},
  year      = {2020},
  url       = {https://arxiv.org/abs/2010.10152},
  eprinttype = {arXiv},
  eprint    = {2010.10152},
  bibsource = {dblp computer science bibliography, https://dblp.org}}

@misc{MLR,
author = "Dua, Dheeru and Graff, Casey",
year = "2017",
title = "{UCI} Machine Learning Repository",
url = "http://archive.ics.uci.edu/ml",
institution = "University of California, Irvine, School of Information and Computer Sciences" }

@article{MacKinlay_Data_snooping,
    author = {Lo, Andrew W. and MacKinlay, A. Craig},
    title = "{Data-Snooping Biases in Tests of Financial Asset Pricing Models}",
    journal = {The Review of Financial Studies},
    volume = {3},
    number = {3},
    pages = {431-467},
    year = {2015},
    month = {04},
    issn = {0893-9454},
    doi = {10.1093/rfs/3.3.431},
    url = {https://doi.org/10.1093/rfs/3.3.431},
    eprint = {https://academic.oup.com/rfs/article-pdf/3/3/431/24416126/030431.pdf},
}

@misc{mrtzvrst,
  author={github},
  title={GitHub},
  year={2020},
  url={https://github.com/mrtzvrst},
}

@ARTICLE{Wei_Fed_DP,  author={Wei, Kang and Li, Jun and Ding, Ming and Ma, Chuan and Yang, Howard H. and Farokhi, Farhad and Jin, Shi and Quek, Tony Q. S. and Poor, H. Vincent},  journal={IEEE Transactions on Information Forensics and Security},   title={Federated Learning With Differential Privacy: Algorithms and Performance Analysis},   year={2020},  volume={15},  number={},  pages={3454-3469},  doi={10.1109/TIFS.2020.2988575}}

@ARTICLE{Rong,
  author={Yu, Rong and Li, Peichun},
  journal={IEEE Network}, 
  title={Toward Resource-Efficient Federated Learning in Mobile Edge Computing}, 
  year={2021},
  volume={35},
  number={1},
  pages={148-155},
  doi={10.1109/MNET.011.2000295}}

@article{Cheng_2020,
author = {Cheng, Yong and Liu, Yang and Chen, Tianjian and Yang, Qiang},
title = {Federated Learning for Privacy-Preserving AI},
year = {2020},
issue_date = {December 2020},
publisher = {Association for Computing Machinery},
address = {New York, NY, USA},
volume = {63},
number = {12},
issn = {0001-0782},
url = {https://doi.org/10.1145/3387107},
doi = {10.1145/3387107},
journal = {Commun. ACM},
month = {nov},
pages = {33–36},
numpages = {4}
}

@InProceedings{HE_IVAN,
author="Damg{\aa}rd, Ivan
and Jurik, Mads",
editor="Kim, Kwangjo",
title="A Generalisation, a Simplification and Some Applications of Paillier's Probabilistic Public-Key System",
booktitle="Public Key Cryptography",
year="2001",
publisher="Springer Berlin Heidelberg",
address="Berlin, Heidelberg",
pages="119--136",
isbn="978-3-540-44586-9"
}

@INPROCEEDINGS{Andrew_SMC,  author={Yao, Andrew C.},  booktitle={23rd Annual Symposium on Foundations of Computer Science (sfcs 1982)},   title={Protocols for secure computations},   year={1982},  volume={},  number={},  pages={160-164},  doi={10.1109/SFCS.1982.38}}

@InProceedings{McMahan_2017,
  title = 	 {{Communication-Efficient Learning of Deep Networks from Decentralized Data}},
  author = 	 {McMahan, Brendan and Moore, Eider and Ramage, Daniel and Hampson, Seth and Arcas, Blaise Aguera y},
  booktitle = 	 {Proceedings of the 20th International Conference on Artificial Intelligence and Statistics},
  pages = 	 {1273--1282},
  year = 	 {2017},
  editor = 	 {Singh, Aarti and Zhu, Jerry},
  volume = 	 {54},
  series = 	 {Proceedings of Machine Learning Research},
  month = 	 {20--22 Apr},
  publisher =    {PMLR},
  url = 	 {https://proceedings.mlr.press/v54/mcmahan17a.html}
}

@article{Andrew_keyboard_prediction,
  author    = {Andrew Hard and
               Kanishka Rao and
               Rajiv Mathews and
               Fran{\c{c}}oise Beaufays and
               Sean Augenstein and
               Hubert Eichner and
               Chlo{\'{e}} Kiddon and
               Daniel Ramage},
  title     = {Federated Learning for Mobile Keyboard Prediction},
  journal   = {CoRR},
  volume    = {abs/1811.03604},
  year      = {2018},
  url       = {http://arxiv.org/abs/1811.03604},
  eprinttype = {arXiv},
  eprint    = {1811.03604},
  bibsource = {dblp computer science bibliography, https://dblp.org}
}

@misc{Francoise_Keyboard,
title	= {Federated Learning for Emoji Prediction in a Mobile Keyboard},
author	= {Francoise Beaufays and Kanishka Rao and Rajiv Mathews and Swaroop Ramaswamy},
year	= {2019},
URL	= {https://arxiv.org/abs/1906.04329}
}

@article{Wenqi_health,
  author    = {Wenqi Li and
               Fausto Milletari and
               Daguang Xu and
               Nicola Rieke and
               Jonny Hancox and
               Wentao Zhu and
               Maximilian Baust and
               Yan Cheng and
               S{\'{e}}bastien Ourselin and
               M. Jorge Cardoso and
               Andrew Feng},
  title     = {Privacy-preserving Federated Brain Tumour Segmentation},
  journal   = {CoRR},
  volume    = {abs/1910.00962},
  year      = {2019},
  url       = {http://arxiv.org/abs/1910.00962},
  eprinttype = {arXiv},
  eprint    = {1910.00962},
  bibsource = {dblp computer science bibliography, https://dblp.org}
}

@INPROCEEDINGS{Songtao_health,  author={Lu, Songtao and Zhang, Yawen and Wang, Yunlong},  booktitle={2020 54th Annual Conference on Information Sciences and Systems (CISS)},   title={Decentralized Federated Learning for Electronic Health Records},   year={2020},  volume={},  number={},  pages={1-5},  doi={10.1109/CISS48834.2020.1570617414}}

@article{Wang_2020,
author = {Wang, Yichuan and Tian, Yuying and Yin, Xinyue and Hei, Xinhong},
year = {2020},
month = {12},
pages = {218-228},
title = {A trusted recommendation scheme for privacy protection based on federated learning},
volume = {3},
journal = {CCF Transactions on Networking},
doi = {10.1007/s42045-020-00045-8}
}

@article{Cuff,
  title={Differential Privacy as a Mutual Information Constraint},
  author={Paul W. Cuff and Lanqing Yu},
  journal={Proceedings of the 2016 ACM SIGSAC Conference on Computer and Communications Security},
  year={2016},
  url={https://api.semanticscholar.org/CorpusID:9204999}
}

@ARTICLE{verdu,
  author={Dongning Guo and Shamai, S. and Verdu, S.},
  journal={IEEE Transactions on Information Theory}, 
  title={Mutual information and minimum mean-square error in Gaussian channels}, 
  year={2005},
  volume={51},
  number={4},
  pages={1261-1282},
  doi={10.1109/TIT.2005.844072}}

@article{Kai_2019,
  author    = {Kai Yang and
               Tao Fan and
               Tianjian Chen and
               Yuanming Shi and
               Qiang Yang},
  title     = {A Quasi-Newton Method Based Vertical Federated Learning Framework
               for Logistic Regression},
  journal   = {CoRR},
  volume    = {abs/1912.00513},
  year      = {2019},
  url       = {http://arxiv.org/abs/1912.00513},
  eprinttype = {arXiv},
  eprint    = {1912.00513},
  bibsource = {dblp computer science bibliography, https://dblp.org}
}

@misc{LR1,
      title={Transparency, Auditability and eXplainability of Machine Learning Models in Credit Scoring}, 
      author={Michael Bücker and Gero Szepannek and Alicja Gosiewska and Przemyslaw Biecek},
      year={2020},
      eprint={2009.13384},
      archivePrefix={arXiv},
      primaryClass={stat.ML},
      url={https://arxiv.org/abs/2009.13384}, 
}

@misc{LR2,
      title={The Impact of Feature Selection and Transformation on Machine Learning Methods in Determining the Credit Scoring}, 
      author={Oguz Koc and Omur Ugur and A. Sevtap Kestel},
      year={2023},
      eprint={2303.05427},
      archivePrefix={arXiv},
      primaryClass={q-fin.RM},
      url={https://arxiv.org/abs/2303.05427}, 
}

@article{LR3,
   title={FedScore: A privacy-preserving framework for federated scoring system development},
   volume={146},
   ISSN={1532-0464},
   url={http://dx.doi.org/10.1016/j.jbi.2023.104485},
   DOI={10.1016/j.jbi.2023.104485},
   journal={Journal of Biomedical Informatics},
   publisher={Elsevier BV},
   author={Li, Siqi and Ning, Yilin and Ong, Marcus Eng Hock and Chakraborty, Bibhas and Hong, Chuan and Xie, Feng and Yuan, Han and Liu, Mingxuan and Buckland, Daniel M. and Chen, Yong and Liu, Nan},
   year={2023},
   month=oct, pages={104485} }

@article{LR4,
  author       = {Shipe, Matthew E. and Deppen, Stephen A. and Farjah, Farhood and Grogan, Eric L.},
  title        = {Developing prediction models for clinical use using logistic regression: an overview},
  journal      = {Journal of Thoracic Disease},
  year         = {2019},
  month        = mar,
  volume       = {11},
  number       = {Suppl 4},
  pages        = {S574--S584},
  doi          = {10.21037/jtd.2019.01.25},
  pmid         = {31032076},
  pmcid        = {PMC6465431}
}
\bibliographystyle{IEEEtran}
\end{document}